\documentclass[12pt,onecolumn]{IEEEtran}

\usepackage[T1]{fontenc}
\usepackage{amsmath,amssymb,amsthm,mathtools}
\usepackage{microtype}
\usepackage{enumitem}
\usepackage{booktabs}
\usepackage{tabularx}
\usepackage{float}
\usepackage{tikz}
\usetikzlibrary{arrows.meta,positioning,calc}
\usepackage{xcolor}
\usepackage[hidelinks]{hyperref}
\usepackage[nameinlink,noabbrev,capitalise]{cleveref}

\allowdisplaybreaks
\newtheorem{theorem}{Theorem}[section]
\newtheorem{proposition}[theorem]{Proposition}
\newtheorem{lemma}[theorem]{Lemma}
\newtheorem{corollary}[theorem]{Corollary}
\theoremstyle{definition}
\newtheorem{definition}[theorem]{Definition}
\theoremstyle{remark}
\newtheorem{remark}[theorem]{Remark}

\crefname{theorem}{theorem}{theorems}
\Crefname{theorem}{Theorem}{Theorems}
\crefname{proposition}{proposition}{propositions}
\Crefname{proposition}{Proposition}{Propositions}
\crefname{lemma}{lemma}{lemmas}
\Crefname{lemma}{Lemma}{Lemmas}
\crefname{corollary}{corollary}{corollaries}
\Crefname{corollary}{Corollary}{Corollaries}
\crefname{definition}{definition}{definitions}
\Crefname{definition}{Definition}{Definitions}
\crefname{remark}{remark}{remarks}
\Crefname{remark}{Remark}{Remarks}

\newcommand{\Z}{\mathbb Z}
\newcommand{\supp}{\operatorname{supp}}
\newcommand{\sgn}{\operatorname{sgn}}
\newcommand{\wtL}{\operatorname{wt}_{\mathrm L}}
\newcommand{\dL}{d_{\mathrm L}}
\newcommand{\dH}{d_{\mathrm H}}
\newcommand{\cE}{\mathcal E}
\newcommand{\cC}{\mathcal C}
\newcommand{\e}{\mathrm e}
\newcommand{\eps}{\varepsilon}
\newcommand{\AL}{A_{\mathrm L}}

\makeatletter
\newcommand{\centeredfigurecaption}[1]{%
  \begingroup
  \long\def\@makecaption##1##2{%
    \@IEEEfigurecaptionsepspace
    \parbox[t]{\hsize}{\centering\normalfont\footnotesize
      ##1.\nobreakspace\nobreakspace ##2\par}}%
  \caption{#1}%
  \endgroup
}
\makeatother

\title{On the Existence of Complete Coverings of \\Lee Spheres}
\author{Jing Zhang and Yue Zhou \thanks{J.\ Zhang is with Key Laboratory of Computing and Stochastic Mathematics (Ministry of Education), School of Mathematics and Statistics, Hunan Normal University, Changsha, Hunan 410081, P.~R.~China (e-mail: zhangjing@hunnu.edu.cn). Y.\ Zhou is with the College of Science, National University of Defense Technology, Changsha 410073, P.~R.~China (e-mail: yue.zhou.ovgu@gmail.com). Y.\ Zhou is supported by the National Natural Science Foundation of China (No.~12371337).}}
\date{\today}

\begin{document}
\maketitle

\begin{abstract}
	Every perfect $r$-error-correcting Lee code in $\mathbb{Z}^n$ induces a finite configuration around each codeword: a family of pairwise disjoint radius-$r$ Lee spheres, all disjoint from the central sphere, that covers every lattice point at Lee distance $r+1$ from the codeword. We call such a configuration a complete covering. Under projection onto a suitable Cartesian product of cycles,	a complete covering corresponds to an optimal local structure in the sense of \v{S}pacapan. The existence of a complete covering is	necessary but, in general, not sufficient for the existence of a perfect Lee code, as our construction for $(n,r)=(3,2)$ shows.
	
	We prove that no complete covering exists whenever $n\ge 8$ and $r^2+\lfloor 25r/34\rfloor-3\ge 2n$.
	For the corresponding Cartesian products of cycles, this strengthens \v{S}pacapan's nonexistence theorem by replacing the linear radius condition $r\ge n$ with a threshold of order $\sqrt{2n}$. We also obtain an improved large-radius nonexistence bound for
	perfect Lee codes, strengthening the bound of Horak and Kim. Furthermore, we show that, for every $\varepsilon>0$, no complete
	covering exists when $n\ge 3$, $r$ is sufficiently large, and $r^2+(2/\mathrm{e}-\varepsilon)r\ge 2n$.
	
	On the positive side, we use near-optimal ternary constant-weight codes to construct almost complete coverings whose pairwise
	disjoint shadows cover a $1-o(1)$ proportion of the first outer	Lee shell as the dimension tends to infinity, for every fixed radius.
\end{abstract}


\section{Introduction}\label{sec:intro}

For $x=(x_1,\ldots,x_n)\in\Z^n$, its \emph{Lee weight} is
\[
 \wtL(x)=\sum_{i=1}^n |x_i|,
\]
and the \emph{Lee distance} between $x,y\in\Z^n$ is
\[
 \dL(x,y)=\wtL(x-y).
\]
The \emph{support} of $x$ is
\[
 \supp(x)=\{i\in\{1,\ldots,n\}:x_i\ne0\}.
\]
The \emph{Lee metric} was introduced for nonbinary error correction by Lee \cite{Lee1958}.  It has appeared in coding problems for constrained and partial-response channels \cite{RothSiegel1994}, multidimensional cluster and burst errors \cite{BlaumBruckVardy1998,EtzionYaakobi2009}, flash memories \cite{Schwartz2012}, and resource placement in torus-based computer networks \cite{BaeBose1997}.  More recent work connects Lee-metric codes with deletion and repetition channels \cite{SokEtAl2018}, code-based cryptography \cite{HorlemannTrautmannWeger2021}, iterative decoding over finite integer rings for Lee channels \cite{BariffiEtAl2022}, and list decoding of quaternary codes \cite{GreferathZumbragel2022}.  These developments continue to motivate the study of structural and extremal questions in the Lee metric.

For $t\ge0$, the set
\[
 S^n(t)=\{x\in\Z^n:\wtL(x)=t\}
\]
is the \emph{Lee shell} of radius $t$.  Following the usual coding-theoretic terminology, we call
\[
 B^n(r)=\{x\in\Z^n:\wtL(x)\le r\}
\]
the \emph{Lee sphere} of radius $r$, although it is the closed ball in metric terminology.  A set $C\subseteq\Z^n$ is a \emph{perfect $r$-error-correcting Lee code} if
\[
 \Z^n=\bigsqcup_{c\in C}\bigl(B^n(r)+c\bigr).
\]
Equivalently, the Lee spheres centered at the codewords form a translational tiling of $\Z^n$.

In \cite{GolombWelch1970}, Golomb and Welch constructed perfect Lee codes under the conditions that $n=1,2$ or $r=1$. In the same paper, they also established the nonexistence of such tilings for sufficiently large radii (with a threshold $r_n$ depending on $n$, not explicitly specified), and conjectured that no perfect $r$-error-correcting Lee code exists for
\[
 n\ge3,
 \qquad
 r\ge2.
\]
This is the famous Golomb-Welch conjecture.

Let $\Z_q$ denote the ring of integers modulo $q$. A similar conjecture for tilings of $\Z_q^n$ by Lee balls was proposed in \cite[Section 3]{GolombWelch1970}. In fact, for any $q\geq 2r+1$, a tiling of $\Z_q^n$ by Lee balls of radius $r$ lifts to a $q$-periodic tiling of $\Z^n$. In the literature, the conjecture for $\Z_q^n$ is usually called the weak version and the one for $\Z^n$ is called the strong version.

Both the weak and strong Golomb--Welch conjectures remain open. Many partial results are known: for dimension $n=3$, they have been solved in \cite{GravierMollardPayan1998}.  For large radii, Post \cite{Post1975} obtained a linear bound for perfect Lee codes in $\Z_q^n$, and Lepist\"o used a modified Elias argument to obtain a bound of square-root order in \cite{Lepisto1981}.  Horak and Kim removed periodicity from the relevant large-radius arguments and combined them with a linear-programming method; in particular, they proved nonexistence for $n\ge876$ and $r\ge\sqrt{2n+40}$ \cite{HorakKim2018}.

Association-scheme methods form another important direction.  Sol\'e developed the \emph{Lee association scheme} and \emph{Lloyd-type conditions} for perfect codes \cite{Sole1988,Sole1989,Sole1991}; recently, Shi, Wang, and Sol\'e returned to this viewpoint and obtained further nonexistence results in a broader weakly metric association-scheme framework \cite{ShiWangSole2026}.

A complementary line of work concerns \emph{linear perfect Lee codes}, equivalently perfect Lee codes whose codewords form a lattice.  Algebraic and number-theoretic methods have produced nonexistence results for fixed radii and infinite families of dimensions \cite{HorakGrosek2014,Kim2017,QureshiCampelloCosta2018,ZhangGe2017,ZhangZhou2019,Qureshi2020}.  In particular, Leung and the second author proved that no linear perfect Lee code of radius two exists in any dimension $n\ge3$ \cite{LeungZhou2020}.  

For more results on the Golomb-Welch conjecture, see the survey \cite{HorakKim2018} by Horak and Kim.


The purpose of this paper is to isolate and study a necessary local condition which every perfect Lee code must satisfy.  Let
\[
 \cE_n=\{\pm e_1,\ldots,\pm e_n\}
\]
be the set of \emph{signed Lee unit vectors}.

\begin{definition}\label{def:complete-covering}
A set $\Omega\subseteq\Z^n$ is a \emph{complete covering of $B^n(r)$ by its translates} if
\begin{enumerate}[label=\textup{(\roman*)}]
 \item $\wtL(z)=2r+1$ for every $z\in\Omega$;
 \item for every $y\in S^n(r)$ and every $u\in\cE_n$ such that $y+u\notin B^n(r)$, there is exactly one $z\in\Omega$ for which
 \[
  y+u\in B^n(r)+z.
 \]
\end{enumerate}
\end{definition}

The signed-unit convention in part~\textup{(ii)} makes the condition apply to the entire first outer shell.  Every perfect Lee code in $\Z^n$ induces a complete covering around each codeword.  The converse is false: in Section~\ref{sec:constructions} we give a complete covering of $B^3(2)$ which cannot be extended across the next shell to a perfect Lee code.  We also classify the complete coverings in dimension two and prove that every one of them extends to a linear perfect Lee code.

For comparison with asymptotic coding constructions, we use the following relaxation in  which the \emph{shadow on the first outer shell} of $z\in S^n(2r+1)$ is defined as
\[
D(z)=S^n(r+1)\cap\bigl(B^n(r)+z\bigr).
\]

\begin{definition}\label{def:almost-covering}
For a fixed positive integer $r$, a sequence $\Omega_n\subseteq S^n(2r+1)$ is an \emph{almost complete covering} if its associated shadows on the first outer shell are pairwise disjoint and cover a $1-o(1)$ proportion of $S^n(r+1)$ as $n\to\infty$.
\end{definition}

Our two main results exclude complete coverings in a range close to the quadratic boundary $2n=r^2$.

\begin{theorem}\label{thm:main}
Let $n\ge8$ and $r\ge2$.  If
\begin{equation}\label{eq:main-condition}
 r^2+\left\lfloor\frac{25r}{34}\right\rfloor-3\ge2n,
\end{equation}
then no complete covering of $B^n(r)$ exists.  Consequently, no perfect $r$-error-correcting Lee code in $\Z^n$ of length $n$ exists.
\end{theorem}

The next result provide a slight improvement to the quadratic boundary under the assumption that $r$ is not smaller than $r_0(\eps)$ which is independent from $n$.

\begin{theorem}\label{thm:asymptotic-intro}
For every $\eps>0$, there exists $r_0(\eps)$ such that, for $n\ge3$, no complete covering of $B^n(r)$ exists whenever $r\ge r_0(\eps)$ and
\[
 r^2+\left(\frac2{\e}-\eps\right)r\ge2n.
\]
Consequently, no perfect $r$-error-correcting Lee code in $\Z^n$ of length $n$ exists in this range.
\end{theorem}

Since the existence of a complete covering is a necessary—but not sufficient—condition for a perfect Lee code, our local impossibility results are stronger than merely ruling out perfect codes. It also yields a strict improvement over the large-radius perfect-code bound highlighted in \cite{HorakKim2018}.  The comparison is summarized in Table~\ref{tab:large-radius-comparison}.

\begin{table}[H]
\centeredfigurecaption{Comparison with the large-radius bounds}
\label{tab:large-radius-comparison}
\centering
\renewcommand{\arraystretch}{1.16}
\begin{tabularx}{\textwidth}{@{}p{0.25\textwidth}X p{0.3\textwidth}@{}}
\toprule
Source & Parameter condition & Conclusion \\
\midrule
Lepist\"o \cite{Lepisto1981} for $\mathbb{Z}^n_q$, extended to $\mathbb{Z}_q^n$ in  \cite{HorakKim2018} & $r\geq 285$ and $r^2+4r+4 > 2.1 n$ & No perfect Lee code\\
Horak--Kim \cite{HorakKim2018}
& $n\ge876$ and $r^2-40\ge2n$
& No perfect Lee code \\
Theorem~\ref{thm:main}
& $n\ge8$ and $r^2+\lfloor25r/34\rfloor-3\ge2n$
& No complete covering; hence no perfect Lee code \\
Theorem~\ref{thm:asymptotic-intro}
& $n\ge3$, $r$ sufficiently large, and $r^2+(2/\mathrm e-\eps)r\ge2n$
& No complete covering; hence no perfect Lee code \\
\bottomrule
\end{tabularx}
\end{table}

Indeed, the Horak--Kim range is contained in the range of Theorem~\ref{thm:main}: if $r^2\ge2n+40$, then the left-hand side of \eqref{eq:main-condition} is larger than $2n$.  The containment is strict.  For example, Theorem~\ref{thm:main} excludes a perfect Lee code with $(n,r)=(876,42)$, whereas $42<\sqrt{2\cdot876+40}$.  For large $n$, the threshold in Theorem~\ref{thm:main} has the form
\[
 r\ge \sqrt{2n}-\frac{25}{68}+O(n^{-1/2}),
\]
and the asymptotic threshold in Theorem~\ref{thm:asymptotic-intro} approaches
\[
 r\ge \sqrt{2n}-\frac1{\mathrm e}+o(1).
\]
By contrast, $\sqrt{2n+40}=\sqrt{2n}+O(n^{-1/2})$.  Thus the present results give new nonexistence cases for perfect Lee codes and, at the same time, rule out the stronger local boundary configuration.

\begin{remark}\label{rem:lepisto-horak-kim}
It is curious that our method gives almost the same bound as Lepist\"o's and the Horak-Kim bound. However, the hypotheses and the methods for these three bounds are quite different. Lepist\"o begins with a perfect Lee code and averages over all translates of a coordinate-truncated shell of radius $r+2$ in order to find a critical set containing many codewords; a coordinate-distribution estimate then forces two of those codewords to be too close \cite{Lepisto1981}.  The large-radius result highlighted by Horak and Kim is likewise a global perfect-code result: it combines density arguments, including extensions of earlier periodic arguments, with a linear-programming estimate \cite{HorakKim2018}.  Neither argument starts from the first outer shell around one Lee sphere alone.  Since a complete covering need not extend to a perfect Lee code, the nonexistence theorems in \cite{Lepisto1981,HorakKim2018} do not directly imply the nonexistence of complete coverings.

There is nevertheless a useful local resemblance.  For every $a\in S^n(r-1)$ and every center $z\in\Omega$ satisfying
$a\preceq z$, the vector $z-a$ has Lee weight $r+2$, which is the same shell radius used in Lepist\"o's critical sets.  In our setting, however, exact coverage gives more rigid information: every compatible pair of outward directions and every repeated outward direction is assigned to exactly one center.  We use this exact local decomposition to obtain an incidence-rank lower bound, and we combine it with a coefficient-gap upper bound derived solely from the first outer shell.  Thus the present proof does not require the global averaging step that is essential in the modified Elias approach.
\end{remark}

There is also a coding interpretation.  Every complete covering is a \emph{constant-Lee-weight code} of weight $2r+1$ and minimum Lee distance at least $2r+2$ which attains equality in a natural shadow-packing inequality.  Theorems~\ref{thm:main} and~\ref{thm:asymptotic-intro} show that this equality cannot occur in the stated parameter ranges.  On the positive side, the near-optimal ternary constant-weight codes of Liu and Shangguan \cite{LiuShangguan2025} give almost complete coverings for every fixed $r$ and all sufficiently large $n$.  This construction asymptotically attains the corresponding Johnson-type bound on the unit-coordinate part of the Lee shell.

The paper is organized as follows.  Section~\ref{sec:prelim} gives the shadow formulation, its coding interpretation, the relation with optimal local structures in Cartesian products of cycles, and the generating-function and rank tools.  Section~\ref{sec:constructions} contains the exact low-dimensional constructions and the almost-complete construction.  Section~\ref{sec:nonexistence} develops local decompositions associated with compatibility graphs and proves Theorems~\ref{thm:main} and~\ref{thm:asymptotic-intro}.  Section~\ref{sec:conclusion} summarizes the consequences and records some further questions.

\section{Preliminaries}\label{sec:prelim}

\subsection{Signed domination and exact shadow partition}

For $x,z\in\Z^n$, write $x\preceq z$ if, for every coordinate $i$,
\[
 x_i z_i\ge0
 \qquad\text{and}\qquad
 |x_i|\le |z_i|.
\]
We refer to this relation as \emph{signed coordinatewise domination}.  Thus every coordinate of $x$ lies on the integer interval joining $0$ to the corresponding coordinate of $z$.

For $z\in S^n(2r+1)$, its \emph{shadow on the first outer shell} is
\[
 D(z)=S^n(r+1)\cap\bigl(B^n(r)+z\bigr).
\]
For convenience, we usually call it the shadow of $z$ when the radius $r$ is clear from the context.
The following elementary formula will be used throughout the paper.

\begin{lemma}\label{lem:shadow}
For every $z\in S^n(2r+1)$,
\[
 D(z)=\{x\in S^n(r+1):x\preceq z\}.
\]
\end{lemma}

\begin{proof}
If $x\in D(z)$, then $\wtL(x)=r+1$ and $\wtL(z-x)\le r$.  Hence
\[
 2r+1=\wtL(z)
 \le \wtL(x)+\wtL(z-x)
 \le(r+1)+r=2r+1.
\]
Equality holds throughout.  Equality in the $\ell_1$ triangle inequality holds coordinatewise and says that there is no cancellation between $x_i$ and $z_i-x_i$.  This is equivalent to $x_i z_i\ge0$ and $|x_i|\le|z_i|$ for every $i$, so $x\preceq z$.

Conversely, if $x\in S^n(r+1)$ and $x\preceq z$, then there is no coordinate cancellation in $z=x+(z-x)$.  Therefore
\[
 \wtL(z-x)=\wtL(z)-\wtL(x)=r,
\]
and hence $x\in B^n(r)+z$.
\end{proof}

The definition of a complete covering is therefore equivalent to an exact partition of one Lee shell.

\begin{proposition}\label{prop:shell-partition}
A set $\Omega\subseteq S^n(2r+1)$ is a complete covering if and only if
\begin{equation}\label{eq:shell-partition}
 S^n(r+1)=\bigsqcup_{z\in\Omega}D(z).
\end{equation}
\end{proposition}

\begin{proof}
Every $x\in S^n(r+1)$ can be written as $x=y+u$, where $y\in S^n(r)$, $u\in\cE_n$, and $y+u\notin B^n(r)$: choose a nonzero coordinate of $x$ and subtract one unit in its sign.  Part~\textup{(ii)} of Definition~\ref{def:complete-covering} therefore says precisely that every point of $S^n(r+1)$ lies in exactly one shadow.  The converse is immediate.
\end{proof}

In Figure \ref{fig:complete_covering}, we provide an example in $\mathbb{Z}^2$ to illustrate Proposition \ref{prop:shell-partition}. In  this example, $r=n=2$, $S^2(3)$ is colored in yellow, $\Omega$ is colored in blue and it consists of points $(3,-2), (2,3), (-3,2)$ and  $(-2,-3)$. 
\begin{figure}[h!]
	\centering
	\includegraphics[width=0.5\textwidth]{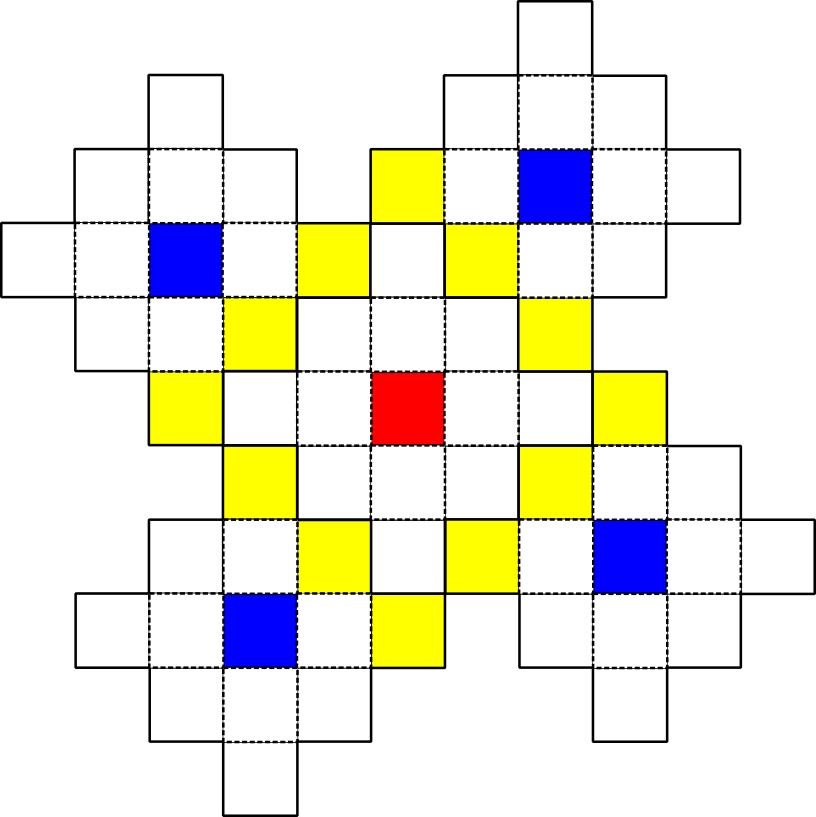}
	\centeredfigurecaption{A complete covering $\Omega\subseteq S^2(5)$ and the corresponding partition of $S^2(3)$\label{fig:complete_covering}}
\end{figure}

For later comparison with tilings and codes, we record the separation forced by the shell partition presented in Proposition~\ref{prop:shell-partition}

\begin{lemma}\label{lem:disjoint-balls}
If $\Omega$ is a complete covering, then the balls
\[
 \{B^n(r)+z:z\in\Omega\}
\]
are pairwise disjoint, are disjoint from $B^n(r)$, and cover $S^n(r+1)$.
\end{lemma}

\begin{proof}
The covering assertion follows from Proposition~\ref{prop:shell-partition}.  Since $\wtL(z)=2r+1$, the central ball and $B^n(r)+z$ are disjoint.

For distinct $z,z'\in S^n(2r+1)$, put
\[
 \mu(z,z')=
 \sum_{i:z_i z_i'>0}\min\{|z_i|,|z_i'|\}.
\]
By Lemma~\ref{lem:shadow}, the two shadows intersect if and only if $\mu(z,z')\ge r+1$.  A coordinatewise calculation gives
\begin{equation}\label{eq:distance-mu}
 \dL(z,z')=4r+2-2\mu(z,z').
\end{equation}
The shadows in a complete covering are disjoint, so $\mu(z,z')\le r$, and hence $\dL(z,z')\ge2r+2$.  Thus the corresponding radius-$r$ balls are disjoint.
\end{proof}

\subsection{Optimal local structures}\label{sec:optimal-local}

\v{S}pacapan studied a closely related local problem in Cartesian products of cycles and paths \cite{Spacapan2007}.  Let $G$ be a graph with graph distance $d_G$.  For a vertex $v$, write
\[
 B_G(v,r)=\{x\in V(G):d_G(x,v)\le r\}
\]
and
\[
 N_G(v,r)=\{x\in V(G):d_G(x,v)=r+1\}.
\]
An \emph{optimal local structure} for $B_G(v,r)$ is a family of pairwise disjoint radius-$r$ balls, each disjoint from $B_G(v,r)$, whose union contains $N_G(v,r)$.  This definition isolates the boundary configuration forced around one codeword of a perfect code in a graph.  In $\Z^n$, Lemma~\ref{lem:disjoint-balls} shows that every complete covering gives an optimal local structure around the central Lee sphere.  Conversely, after discarding balls which do not meet the outer shell, an optimal local structure in $\Z^n$ has centers of Lee weight $2r+1$ and gives a complete covering.  Thus the two notions coincide in the integer lattice.

The following proposition gives an explicit transfer in both directions between the integer lattice and sufficiently long Cartesian products of cycles.  The two directions require different lower bounds on the cycle lengths because descending a complete configuration must preserve all neighboring balls simultaneously, whereas lifting uses only distances from the central ball and its first outer shell.

\begin{proposition}\label{prop:cycle-lifting}
Let
\[
 G=C_{\ell_1}\times \cdots\times C_{\ell_n}.
\]
The following statements hold.
\begin{enumerate}[label=\textup{(\roman*)}]
 \item If $B^n(r)$ admits a complete covering and
 \[
  \ell_i\ge6r+3\qquad(1\le i\le n),
 \]
 then a radius-$r$ ball in $G$ admits an optimal local structure.
 \item If a radius-$r$ ball in $G$ admits an optimal local structure and
 \[
  \ell_i\ge4r+3\qquad(1\le i\le n),
 \]
 then $B^n(r)$ admits a complete covering in $\Z^n$.
\end{enumerate}
\end{proposition}

\begin{proof}
Let
\[
 \pi:\Z^n\longrightarrow
 C_{\ell_1}\times \cdots\times C_{\ell_n}
\]
be coordinatewise reduction.

For part~\textup{(i)}, let $\Omega$ be a complete covering.  Every point in the central ball or in one of the balls $B^n(r)+z$, $z\in\Omega$, has every coordinate in the interval
\[
 [-3r-1,3r+1].
\]
When $\ell_i\ge6r+3$, reduction is injective on this box.  Moreover, every radius-$r$ ball in the product of cycles is the image of the corresponding Lee ball in $\Z^n$, and the shell $N_G(0,r)$ is the image of $S^n(r+1)$; these assertions follow because the relevant radii are strictly smaller than one half of every cycle length.  Lemma~\ref{lem:disjoint-balls} therefore shows that the balls centered at $\pi(z)$, $z\in\Omega$, are pairwise disjoint, avoid the central ball, and cover $N_G(0,r)$.  They form an optimal local structure.

For part~\textup{(ii)}, use vertex transitivity to take the central vertex to be $0$, and discard any ball in the optimal local structure which does not meet $N_G(0,r)$.  If a remaining ball has center $z$, then disjointness from the central ball gives
\[
 d_G(0,z)\ge2r+1,
\]
while its intersection with $N_G(0,r)$ gives the reverse inequality.  Hence $d_G(0,z)=2r+1$.  Since $\ell_i\ge4r+3$, both $z$ and every point $x\in N_G(0,r)$ have unique short representatives $\widetilde z,\widetilde x\in\Z^n$ satisfying
\[
 \wtL(\widetilde z)=2r+1,
 \qquad
 \wtL(\widetilde x)=r+1.
\]
Suppose that the ball centered at $z$ contains $x$.  If a coordinate difference between $\widetilde z$ and $\widetilde x$ wrapped around its cycle, then
\[
 |\widetilde z_i-\widetilde x_i|
 \ge \ell_i-r
 \ge3r+3.
\]
On the other hand,
\[
 |\widetilde z_i-\widetilde x_i|
 \le |\widetilde z_i|+|\widetilde x_i|
 \le3r+2,
\]
a contradiction.  Therefore
\[
 \wtL(\widetilde z-\widetilde x)=d_G(z,x)\le r.
\]
It follows that
\[
 2r+1=\wtL(\widetilde z)
 \le\wtL(\widetilde x)+\wtL(\widetilde z-\widetilde x)
 \le2r+1.
\]
Thus $\wtL(\widetilde z)
=\wtL(\widetilde x)+\wtL(\widetilde z-\widetilde x)$ which implies $\widetilde x\preceq\widetilde z$.  Pairwise disjointness of the balls gives uniqueness.  Hence the lifted centers form a complete covering of $B^n(r)$.
\end{proof}

The following theorem is the main local nonexistence result of \cite{Spacapan2007}.

\begin{theorem}\label{thm:spacapan}
Let
\[
 G=C_{\ell_1}\times \cdots\times C_{\ell_n}.
\]
If
\[
 n\ge3,
 \qquad
 r\ge n,
 \qquad
 \ell_i\ge2r+1\quad(1\le i\le n),
\]
then no radius-$r$ ball in $G$ admits an optimal local structure.
\end{theorem}

By Proposition~\ref{prop:cycle-lifting}\textup{(i)}, the next result follows immediately.

\begin{corollary}\label{cor:spacapan-Zn}
For $n\ge3$ and $r\ge n$, no complete covering of $B^n(r)$ exists.
\end{corollary}

\subsection{Constant-Lee-weight codes and the shadow-packing bound}\label{sec:lee-code-link}

A \emph{constant-Lee-weight code} with parameters $(n,d,w)$ is a set
\[
 \cC\subseteq S^n(w)
\]
such that $\dL(c,c')\ge d$ for all distinct $c,c'\in\cC$.  We write
\[
 \AL(n,d,w)
\]
for the maximum possible cardinality.  Here ``optimal'' in the usual coding-theoretic sense means that the cardinality is $\AL(n,d,w)$.  This should be distinguished from equality in the weighted bound below, because different centers may have shadows of different sizes.

The distance formula in \eqref{eq:distance-mu} gives an exact link between constant-Lee-weight codes and complete coverings.

\begin{proposition}\label{prop:lee-shadow-code}
Let $\cC\subseteq S^n(2r+1)$.  The shadows $D(z)$, $z\in\cC$, are pairwise disjoint if and only if $\cC$ is a constant-Lee-weight code with minimum Lee distance at least $2r+2$.  In that case,
\begin{equation}\label{eq:weighted-shadow-bound}
 \sum_{z\in\cC}|D(z)|\le |S^n(r+1)|.
\end{equation}
Equality holds in \eqref{eq:weighted-shadow-bound} if and only if $\cC$ is a complete covering.
\end{proposition}

\begin{proof}
For distinct $z,z'\in S^n(2r+1)$, Lemma~\ref{lem:shadow} and \eqref{eq:distance-mu} show that
\[
 D(z)\cap D(z')=\varnothing
 \quad\Longleftrightarrow\quad
 \mu(z,z')\le r
 \quad\Longleftrightarrow\quad
 \dL(z,z')\ge2r+2.
\]
If the shadows are pairwise disjoint, their total size is at most $|S^n(r+1)|$, proving \eqref{eq:weighted-shadow-bound}.  Equality means that the disjoint shadows partition the shell, which is equivalent to a complete covering by Proposition~\ref{prop:shell-partition}.
\end{proof}

Thus Theorems~\ref{thm:main} and~\ref{thm:asymptotic-intro} can be read as strictness results for the weighted shadow-packing bound.  They do not, by themselves, determine the cardinality parameter $\AL(n,2r+2,2r+1)$.

A point of a Lee shell is called \emph{unit-coordinate} if all its nonzero coordinates belong to $\{1,-1\}$.  If every center of $\cC$ is unit-coordinate, then every shadow has size
\[
 \binom{2r+1}{r+1}
\]
and lies entirely in the unit-coordinate part of $S^n(r+1)$.  Hence
\begin{equation}\label{eq:unit-coordinate-shadow-bound}
 |\cC|\binom{2r+1}{r+1}
 \le 2^{r+1}\binom n{r+1}.
\end{equation}
This is the relevant Johnson-type packing bound for the construction in Section~\ref{sec:almost}.  The result in \cite{LiuShangguan2025} produces ternary constant-weight codes which asymptotically attain the right side of \eqref{eq:unit-coordinate-shadow-bound} under a slightly stronger Hamming-distance condition.

\subsection{Generating functions of Lee shells}\label{sec:shell-generating}

For fixed $n$ and $t\geq 0$, put
\[
 L_t=|S^n(t)|.
\]
We regard $q$ as a formal variable.  For one integer coordinate, the contribution of the value $0$ is $1$, while for every positive integer $j$ there are two values, $j$ and $-j$, of Lee weight $j$.  Therefore
\[
 \sum_{m\in\Z}q^{|m|}
 =1+2q+2q^2+\cdots
 =1+\frac{2q}{1-q}
 =\frac{1+q}{1-q}.
\]
For a vector $x=(x_1,\ldots,x_n)$, the monomial contributed by the coordinates is
\[
 q^{|x_1|}\cdots q^{|x_n|}=q^{\wtL(x)}.
\]
Multiplying the one-coordinate series and then grouping vectors according to their Lee weight gives
\begin{equation}\label{eq:F}
	 F(q):=\sum_{t\ge0}L_tq^t
	 =\sum_{x\in\Z^n}q^{\wtL(x)}
	 =\prod_{i=1}^n\left(\sum_{m\in\Z}q^{|m|}\right)
	 =\left(\frac{1+q}{1-q}\right)^n.
\end{equation}
In particular, the coefficient of $q^t$ in $F(q)$ is exactly $L_t$.

The same numbers can be obtained directly by fixing the support size.  If $t\ge1$ and $x\in S^n(t)$ has exactly $j$ nonzero coordinates, then one first chooses the support in $\binom nj$ ways and the signs in $2^j$ ways.  The positive absolute values on the selected coordinates form an ordered composition of $t$ into $j$ positive parts.  By the standard separator argument, the number of such compositions is $\binom{t-1}{j-1}$.  Hence
\begin{equation}\label{eq:shell-explicit}
 L_t=\sum_{j=1}^{\min\{n,t\}}
 2^j\binom nj\binom{t-1}{j-1}.
\end{equation}

Differentiating \eqref{eq:F} as a formal power series gives
\[
 F'(q)=\frac{2n}{1-q^2}F(q),
\]
or equivalently
\[
 (1-q^2)F'(q)=2nF(q).
\]
The coefficient of $q^t$ on the left is $(t+1)L_{t+1}-(t-1)L_{t-1}$.  Therefore
\begin{equation}\label{eq:recurrence}
 (t+1)L_{t+1}-(t-1)L_{t-1}=2nL_t
 \qquad(t\ge1),
\end{equation}
where $L_{-1}=0$.

We shall also need the sum of the support sizes over a shell.  Define
\[
 K_t=\sum_{x\in S^n(t)}|\supp(x)|.
\]
Equivalently, $K_t$ counts the pairs $(x,i)$ such that $x\in S^n(t)$ and $i\in\supp(x)$.  Indeed,
\[
 K_t
 =\sum_{x\in S^n(t)}\sum_{i=1}^n\mathbf 1_{\{x_i\ne0\}}
 =\sum_{i=1}^n\#\{x\in S^n(t):x_i\ne0\}.
\]
To form the generating function for these pairs, first choose the distinguished coordinate $i$.  It is required to be nonzero, and hence contributes
\[
 \sum_{m\in\Z\setminus\{0\}}q^{|m|}
 =2q+2q^2+\cdots
 =\frac{2q}{1-q}.
\]
Each of the remaining $n-1$ coordinates is unrestricted and contributes $(1+q)/(1-q)$.  Summing over the $n$ choices of $i$ yields
\begin{equation}\label{eq:KGF}
	 \sum_{t\ge0}K_tq^t
	 =\sum_{i=1}^n\sum_{x\in\Z^n}
	   \mathbf 1_{\{x_i\ne0\}}q^{\wtL(x)}
	 =n\frac{2q}{1-q}
	   \left(\frac{1+q}{1-q}\right)^{n-1}
	 =\frac{2nq}{1+q}F(q).
\end{equation}
The last identity follows by dividing $F(q)$ by one factor $(1+q)/(1-q)$ and multiplying by $2q/(1-q)$.

Finally,
\[
 2nF(q)-\frac{2nq}{1+q}F(q)
 =\frac{2n}{1+q}F(q)
 =(1-q)F'(q).
\]
The coefficient of $q^{r-1}$ on the left is $2nL_{r-1}-K_{r-1}$, while the coefficient on the right is $rL_r-(r-1)L_{r-1}$.  Thus
\begin{equation}\label{eq:support-identity}
 2nL_{r-1}-K_{r-1}
 =rL_r-(r-1)L_{r-1}.
\end{equation}

A finite sequence $(c_0,\ldots,c_d)$ is \emph{symmetric unimodal} if
\[
c_j=c_{d-j}
\]
and
\[
c_0\le c_1\le\cdots\le c_{\lfloor d/2\rfloor}.
\]
The following standard fact follows, for example, from the symmetric
chain decomposition of products of finite chains in
\cite{deBruijnTengbergenKruyswijk1951}, or from the more general product
theorem in \cite{ProctorSaksSturtevant1980}.

\begin{lemma}\label{lem:unimodal}
	For nonnegative integers $a_1,\ldots,a_s$, the coefficient sequence of
	\[
	\prod_{i=1}^s(1+q+\cdots+q^{a_i})
	\]
	is symmetric unimodal.
\end{lemma}

\subsection{A matrix rank lemma}
We first recall the elementary graph notation used below.  For a
positive integer $v$, let $K_v$ denote the \emph{complete graph} on
$v$ vertices; thus every unordered pair of distinct vertices is an
edge of $K_v$.  A \emph{matching of size $m$} is a collection of $m$
pairwise vertex-disjoint edges, and we use $M_m$ to denote such a
matching.  If $2m\le v$, then
\[
K_v-M_m
\]
denotes the graph obtained from $K_v$ by deleting the $m$ edges of a
copy of $M_m$, while retaining all $v$ vertices.  Thus $K_v-M_m$ has
\[
\binom{v}{2}-m
\]
edges.

For a graph with vertex set $V$, its adjacency matrix is the
$V$-by-$V$ matrix whose $(u,w)$-entry is $1$ when $u$ and $w$ are
adjacent and is $0$ otherwise. It is clear that the adjacency matrix of a complete graph is an all-one matrix which we denote by $J$.

The following elementary matrix lemma is used to convert the local edge decompositions in Section~\ref{sec:descent} into lower bounds.

\begin{lemma}\label{lem:matching-matrix}
Let $P$ be the adjacency matrix of a matching on a finite vertex set $V$, and let
\[
 D=\operatorname{diag}(d_u:u\in V),
 \qquad d_u\ge1.
\]
Then $D-P$ is positive semidefinite.  If $t$ matching edges $\{u,\bar u\}$ satisfy $d_u=d_{\bar u}=1$, then
\[
 \operatorname{rank}(J+D-P)=
 \begin{cases}
 |V|,&t=0,\\
 |V|-t+1,&t\ge1.
 \end{cases}
\]
In particular, if the matching has $m\ge1$ edges, then
\begin{equation}\label{eq:matching-rank}
 \operatorname{rank}(J+D-P)\ge |V|-m+1.
\end{equation}
If every $d_u\ge2$, then $D-P$ is positive definite.
\end{lemma}

\begin{proof}
After ordering matched vertices consecutively, $D-P$ is a direct sum of positive $1$-by-$1$ blocks and matrices
\[
 \begin{pmatrix}d_u&-1\\-1&d_{\bar u}\end{pmatrix}.
\]
Each $2$-by-$2$ block is positive semidefinite and is singular exactly when $d_u=d_{\bar u}=1$; in that case its kernel is spanned by $e_u+e_{\bar u}$.  Hence the kernel of $D-P$ is spanned by the corresponding $t$ vectors.

Both $J$ and $D-P$ are positive semidefinite.  Their sum has kernel equal to the intersection of their kernels.  On the kernel of $D-P$ which is of $t$-dimension, a vector has the form
\[
 x=\sum_{j=1}^t c_j(e_{u_j}+e_{\bar u_j}),
\]
and
\[
 x^{\mathsf T}Jx=4\left(\sum_{j=1}^t c_j\right)^2.
\]
Thus, for $t\geq 1$, the matrix $J$ reduces the nullity by one, thereby yielding the stated rank.  The remaining assertions follow immediately.
\end{proof}

\section{Constructions of Complete and Almost Complete Coverings}\label{sec:constructions}

In this section, we concentrate on the positive results.  The first two are exact local constructions, while the last one is asymptotic.

\subsection{Complete coverings in dimension two}\label{sec:dim2}

The local condition is completely rigid in dimension two.  With the coordinate axes fixed, there are only two possibilities, and both arise from the classical lattice tiling.

\begin{theorem}\label{thm:intro-dim2}
For every $r\ge1$, and with the coordinate axes fixed, there are exactly two complete coverings of $B^2(r)$.  They correspond to the two cyclic orientations of the boundary, are exchanged by a reflection, and both extend to lattice tilings of $\Z^2$ by $B^2(r)$.
\end{theorem}

\begin{proof}
The shell $S^2(r+1)$ has $4(r+1)$ points and four axial points
\[
 (\pm(r+1),0),
 \qquad
 (0,\pm(r+1)).
\]
Let $z=(z_1,z_2)\in S^2(2r+1)$, and put
\[
 a=|z_1|,
 \qquad
 b=|z_2|.
\]
Then $a+b=2r+1$, so exactly one of $a,b$ is at least $r+1$.  By Lemma~\ref{lem:shadow}, $D(z)$ contains the corresponding axial point and no other axial point.  Hence every shadow contains exactly one axial point, and a complete covering has at most four centers.

If, for example, $a\ge r+1$, then $b\le r$ and
\[
 D(z)=
 \bigl\{(\sgn(z_1)(r+1-j),\sgn(z_2)j):0\le j\le b\bigr\}.
\]
Thus $|D(z)|=b+1\le r+1$.  Since the shadows partition a set of size $4(r+1)$, there are exactly four shadows, every shadow has size $r+1$, and
\[
 \{a,b\}=\{r+1,r\}.
\]
Each shadow consists of one axial point and all interior points on one adjacent quadrant arc.  The four arcs can be oriented only cyclically clockwise or cyclically counterclockwise.  Consequently, the two complete coverings are
\[
 \Omega_r^+
 =\{\pm(r+1,r),\ \pm(-r,r+1)\}
\]
and
\[
 \Omega_r^-
 =\{\pm(r+1,-r),\ \pm(r,r+1)\}.
\]
They are exchanged by a reflection.

It remains to consider extendability. Put \[ N=|B^2(r)|=2r^2+2r+1 \] and define the homomorphism \[ \varphi_r:\Z^2\longrightarrow\Z_N, \qquad \varphi_r(x_1,x_2)=(r+1)x_1+rx_2\pmod N. \] The classical construction of two-dimensional perfect Lee codes shows that the restriction of $\varphi_r$ to $B^2(r)$ is a bijection \cite{GolombWelch1970,HorakGrosek2014}. Consequently, \[ L_r^+=\ker\varphi_r \] is a tiling lattice for $B^2(r)$: every coset of $L_r^+$ contains exactly one element of $B^2(r)$. Moreover, \[ \varphi_r(r+1,r)=(r+1)^2+r^2\equiv0\pmod N \] and \[ \varphi_r(-r,r+1)=0. \] Hence \[ \Omega_r^+ = \{\pm(r+1,r),\ \pm(-r,r+1)\} \subseteq L_r^+. \] The shadows of these four centers have already been shown to partition $S^2(r+1)$. Therefore the lattice tiling associated with $L_r^+$ extends the complete covering $\Omega_r^+$. Finally, reflecting the second coordinate transforms $\Omega_r^+$ into $\Omega_r^-$ and transforms this lattice tiling into one extending $\Omega_r^-$.
\end{proof}

\subsection{A nonextendable complete covering in dimension three}\label{sec:three-two}

The next result shows that the existence of a complete covering as a local condition is genuinely weaker than a tiling once the dimension is larger than two.

\begin{theorem}\label{thm:intro-32}
There is a complete covering of $B^3(2)$ with ten centers which cannot be extended to a tiling of $\Z^3$.
\end{theorem}

\begin{proof}
Define
\[
\begin{aligned}
\Omega=\{&
(0,0,5),(0,0,-5),\\
&(3,-1,-1),(-3,1,-1),(1,-3,1),(-1,3,1),\\
&(2,1,2),(-2,-1,2),(1,2,-2),(-1,-2,-2)
\}.
\end{aligned}
\]
Every center has Lee weight five.  The shadows of three representatives are
\[
 D(0,0,5)=\{(0,0,3)\},
\]
\[
 D(3,-1,-1)
 =\{(3,0,0),(2,-1,0),(2,0,-1),(1,-1,-1)\},
\]
and
\[
\begin{aligned}
D(2,1,2)=\{&(2,1,0),(2,0,1),(1,0,2),\\
&(0,1,2),(1,1,1)\}.
\end{aligned}
\]
The remaining shadows are obtained by the group generated by
\[
 (x_1,x_2,x_3)\longmapsto(-x_1,-x_2,x_3)
\]
and
\[
 (x_1,x_2,x_3)\longmapsto(x_2,x_1,-x_3).
\]
The shell $S^3(3)$ has $6$ points of absolute-value type $(3,0,0)$, $24$ points of type $(2,1,0)$, and $8$ points of type $(1,1,1)$.

The orbit of the first shadow has size two and covers the two axial points on the third axis.  The four shadows in the orbit of $D(3,-1,-1)$ cover the four remaining axial points, eight points of type $(2,1,0)$, and the four full-support sign patterns in which the first two signs are opposite.  The four shadows in the orbit of $D(2,1,2)$ cover the other sixteen points of type $(2,1,0)$ and the four full-support sign patterns in which the first two signs agree.  These sets are disjoint, and their sizes add to
\[
 2+16+20=38=|S^3(3)|.
\]
Hence the shadows partition $S^3(3)$, and $\Omega$ is a complete covering.

To prove nonextendability, consider
\[
 x=(-2,-1,-1)\in S^3(4).
\]
Its distance from every center in $\Omega$ is at least three, so it is not covered by any of the selected radius-two spheres.  Suppose that an additional center $c$ covered $x$, while its
radius-two sphere remained disjoint from the central sphere and from
all the selected spheres.  Put
\[
d=(d_1,d_2,d_3)=c-x.
\]
Since $c$ covers $x$, we have
\[
\wtL(d)=\dL(c,x)\le2.
\]
Consider the two existing centers
\[
a=(-3,1,-1),
\qquad
b=(-1,-2,-2).
\]
We have
\[
x-a=(1,-2,0),
\qquad
x-b=(-1,1,1),
\]
and hence
\[
\dL(x,a)=\dL(x,b)=3.
\]
Since the radius-two sphere centered at $c$ is disjoint from those
centered at $a$ and $b$, we must have
\[
\dL(c,a)\ge5,
\qquad
\dL(c,b)\ge5.
\]
On the other hand,
\[
\dL(c,a)
=\wtL\bigl((x-a)+d\bigr)
\le \wtL(x-a)+\wtL(d)
\le3+2=5,
\]
and similarly,
\[
\dL(c,b)
=\wtL\bigl((x-b)+d\bigr)
\le \wtL(x-b)+\wtL(d)
\le3+2=5.
\]
Therefore equality holds throughout both estimates.  In particular,
\[
\wtL(d)=2.
\]
Equality in
\[
\wtL\bigl((1,-2,0)+d\bigr)
\le \wtL(1,-2,0)+\wtL(d)
\]
means that there is no coordinatewise cancellation.  It follows that
\[
d_1\ge0,
\qquad
d_2\le0.
\]
Similarly, equality in
\[
\wtL\bigl((-1,1,1)+d\bigr)
\le \wtL(-1,1,1)+\wtL(d)
\]
gives
\[
d_1\le0,
\qquad
d_2\ge0,
\qquad
d_3\ge0.
\]
Consequently,
\[
d_1=d_2=0.
\]
Since $\wtL(d)=2$ and $d_3\ge0$, we obtain
\[
d=(0,0,2).
\]
It follows that
\[
c=x+d=(-2,-1,1).
\]
However,
\[
\wtL(c)=4<5,
\]
so the radius-two Lee sphere centered at $c$ intersects the central
sphere.  This contradiction proves that the complete covering cannot
be extended to a tiling.
\end{proof}

The global nonexistence of three-dimensional Lee-sphere tilings for every radius at least two was proved in \cite{GravierMollardPayan1998}.  The preceding argument is stronger for this particular local configuration, since the failure is already visible on the next shell.

\subsection{Almost complete coverings from ternary constant-weight codes}\label{sec:almost}

We identify the ternary alphabet with $\{0,+1,-1\}$.  The \emph{Hamming weight} of a word is the number of its nonzero coordinates, and its \emph{Hamming distance} from another word is the number of coordinates in which the two words differ.  A $Q$-ary constant-weight code of length $n$, weight $w$, and minimum Hamming distance $d$ is a code in which every word has Hamming weight $w$ and every two distinct words have Hamming distance at least $d$.  Let $A_Q(n,d,w)$ denote its maximum possible size.

The following asymptotically sharp result is proved in \cite{LiuShangguan2025}.

\begin{theorem}\label{thm:ternary-constant-weight}
Fix $Q,w,d$, where $d$ is odd, and put
\[
 t=\frac{2w-d+1}{2}.
\]
Then, as $n\to\infty$,
\begin{equation}\label{eq:ternary-constant-weight}
 A_Q(n,d,w)
 =(1-o(1))
 \frac{(Q-1)^t}{\binom wt}
 \binom nt.
\end{equation}
\end{theorem}

We first record why the relevant ternary codes give disjoint Lee shadows.  If two such shadows met, the two centers would agree with the same nonzero sign on at least $r+1$ coordinates.  Since both centers have Hamming weight $2r+1$, this would force their Hamming distance to be at most $2r$.

\begin{lemma}\label{lem:code-shadows-disjoint}
Let $c,c'\in\{0,\pm1\}^n$ have Hamming weight $2r+1$.  If
\[
 \dH(c,c')\ge2r+1,
\]
then $D(c)\cap D(c')=\varnothing$.
\end{lemma}

The next theorem gives the almost complete coverings announced in the introduction.

\begin{theorem}\label{thm:intro-almost}
Fix $r\ge1$.  For all sufficiently large $n$, there is a set
\[
 \Omega_n\subseteq\{0,\pm1\}^n\cap S^n(2r+1)
\]
whose shadows are pairwise disjoint and cover a $1-o(1)$ proportion of $S^n(r+1)$.  Moreover,
\begin{equation}\label{eq:almost-size}
 |\Omega_n|
 =(1-o(1))
 \frac{2^{r+1}}{\binom{2r+1}{r+1}}
 \binom n{r+1}.
\end{equation}
\end{theorem}

\begin{proof}
Apply Theorem~\ref{thm:ternary-constant-weight} with
\[
 Q=3,
 \qquad
 w=2r+1,
 \qquad
 d=2r+1.
\]
Then
\[
 t=\frac{2(2r+1)-(2r+1)+1}{2}=r+1.
\]
Hence there is a ternary constant-weight code $\cC_n$ such that
\[
 |\cC_n|
 =(1-o(1))
 \frac{2^{r+1}}{\binom{2r+1}{r+1}}
 \binom n{r+1}.
\]
Regard the codewords as integer vectors and set $\Omega_n=\cC_n$.  Every center has Lee weight $2r+1$, and Lemma~\ref{lem:code-shadows-disjoint} shows that the shadows are pairwise disjoint.

A unit-coordinate center of weight $2r+1$ dominates exactly the unit-coordinate points obtained by choosing $r+1$ of its signed nonzero coordinates.  Hence its shadow has size
\[
 \binom{2r+1}{r+1}.
\]
It follows that
\[
 \left|\bigcup_{z\in\Omega_n}D(z)\right|
 =(1-o(1))2^{r+1}\binom n{r+1}.
\]
This asymptotically attains the unit-coordinate shadow bound \eqref{eq:unit-coordinate-shadow-bound}.

Finally, \eqref{eq:shell-explicit} gives
\[
 |S^n(r+1)|
 =\sum_{j=1}^{r+1}
 2^j\binom nj\binom r{j-1}.
\]
The support-$(r+1)$ term, which counts the unit-coordinate points, is
\[
 2^{r+1}\binom n{r+1}.
\]
For fixed $r$, all terms of smaller support sum to $O_r(n^r)$.  Therefore
\[
 |S^n(r+1)|
 =2^{r+1}\binom n{r+1}+O_r(n^r),
\]
and the selected shadows cover a $1-o(1)$ proportion of the entire outer shell.
\end{proof}

\section{Nonexistence of Complete Coverings}\label{sec:nonexistence}
We now prove the two main nonexistence theorems.  Assume throughout
this section that $\Omega$ is a complete covering.  For every
$a\in S^n(r-1)$, define
\begin{equation}\label{eq:Omega-a}
	\Omega(a)=\{z\in\Omega:a\preceq z\}.
\end{equation}
Thus $\Omega(a)$ is the set of centers in $\Omega$ which dominate $a$

The idea starts with the following equation obtained via double-counting:
\begin{equation}\label{eq:key_idea}
	\sum_{a\in S^n(r-1)} |\Omega(a)| = \sum_{z\in \Omega} |\{a\in S^n(r-1): a\preceq z\}|.
\end{equation}

\begin{figure}[h!]
	\centering
	\resizebox{0.90\textwidth}{!}{%
		\begin{tikzpicture}[
			box/.style={draw, rounded corners, align=center, inner sep=4pt,
				text width=0.30\textwidth, font=\footnotesize},
			wide/.style={draw, rounded corners, align=center, inner sep=4pt,
				text width=0.34\textwidth, font=\footnotesize},
			arr/.style={-{Latex[length=2mm]}, thick}
			]
			\node[wide] (cover) {Complete covering};
			\node[box, below left=8mm and -7mm of cover] (local)
			{Exact local decomposition of\\
				$E(\Gamma_a)\sqcup\{\ell_u:u\in V(a)\}$};
			\node[box, below right=8mm and -7mm of cover] (coeff)
			{Coefficient polynomials $P_z(q)$\\
				and central gaps $g(z)$};
			\node[box, below=6mm of local] (lower)
			{Incidence-rank lower bound for \eqref{eq:key_idea}};
			\node[box, below=6mm of coeff] (upper)
			{Coefficient-gap upper bound for \eqref{eq:key_idea}};
			\coordinate (mid) at ($(lower)!0.5!(upper)$);
			\node[wide, below=10mm of mid] (compare)
			{Comparison gives necessary inequalities\\
				among $L_{r-1}$, $L_r$, and $L_{r+1}$};
			\node[box, below left=7mm and -7mm of compare] (finite)
			{Lee-shell recurrence\\and finite estimates};
			\node[box, below right=7mm and -7mm of compare] (asym)
			{Lee-shell recurrence\\and top-support asymptotics};
			\node[box, below=6mm of finite] (main)
			{Contradiction under the condition\\of Theorem~\ref{thm:main}};
			\node[box, below=6mm of asym] (asymthm)
			{Contradiction under the condition\\of Theorem~\ref{thm:asymptotic-intro}};
			\draw[arr] (cover) -- (local);
			\draw[arr] (cover) -- (coeff);
			\draw[arr] (local) -- (lower);
			\draw[arr] (coeff) -- (upper);
			\draw[arr] (lower) -- (compare);
			\draw[arr] (upper) -- (compare);
			\draw[arr] (compare) -- (finite);
			\draw[arr] (compare) -- (asym);
			\draw[arr] (finite) -- (main);
			\draw[arr] (asym) -- (asymthm);
		\end{tikzpicture}%
	}
	\centeredfigurecaption{Framework of the proofs of Theorems~\ref{thm:main} and~\ref{thm:asymptotic-intro}.\label{fig:proof-framework}}
\end{figure}

The framework of the proofs is shown in Fig.~\ref{fig:proof-framework} which has two complementary parts.
For $a\in S^n(r-1)$, let $V(a)$ denote the signed directions which increase the Lee weight from $a$, and let $\Gamma_a$ be the graph on $V(a)$ whose edges are the compatible pairs of distinct directions.  We also keep one formal loop $\ell_u$ at every $u\in V(a)$.  The precise definitions are given in Subsection~\ref{sec:descent}. By considering the decomposition of this graph, we can derive lower bounds for $|\Omega(a)|$.
On the other hand, the coefficient polynomial $P_z(q)$ associated with each center $z\in \Omega$ and some further analysis give an upper bound for the right-hand-side of \eqref{eq:key_idea}. 

Comparing these two bounds produces necessary inequalities among the three consecutive Lee-shell sizes $L_{r-1}$, $L_r$, and $L_{r+1}$.  The Lee-shell recurrence, together with either finite estimates or the top-support asymptotics, shows that these necessary inequalities cannot hold in the parameter ranges of Theorems~\ref{thm:main} and~\ref{thm:asymptotic-intro}. 

\subsection{The compatibility graph and its local decomposition}\label{sec:descent}

This subsection establishes the first step in the left branch of Fig.~\ref{fig:proof-framework}. 
Assume throughout this subsection that $\Omega$ is a complete covering, and fix
\[
 a\in S^n(r-1).
\]
For each nonzero coordinate $a_i$, exactly one signed unit direction increases the Lee weight, namely $\sgn(a_i)e_i$.  For each zero coordinate, both $e_i$ and $-e_i$ increase the Lee weight.  Define
\[
 V(a)=
 \{\sgn(a_i)e_i:a_i\ne0\}
 \cup
 \{e_i,-e_i:a_i=0\}.
\]
Then
\[
 v_a:=|V(a)|=2n-|\supp(a)|.
\]
For every zero coordinate $a_i=0$, the two directions $e_i$ and
$-e_i$ form an opposite pair.  The number of such pairs is
\[
m_a=n-|\supp(a)|.
\]
These $m_a$ opposite pairs are pairwise vertex-disjoint and hence form
a matching on $V(a)$.

Let $\Gamma_a$ be the \emph{compatibility graph} with vertex set
$V(a)$.  Two distinct vertices $u,w\in V(a)$ are adjacent in
$\Gamma_a$ if and only if they do not form one of the opposite pairs
$\{e_i,-e_i\}$ arising from a zero coordinate of $a$.  Consequently,
\begin{equation}\label{eq:Gamma-structure}
	\Gamma_a\cong K_{v_a}-M_{m_a}.
\end{equation}
Here $K_{v_a}$ is the complete graph on the $v_a$ signed outward
directions in $V(a)$, while $M_{m_a}$ is the matching consisting
precisely of the $m_a$ opposite pairs
\[
\{e_i,-e_i\},
\qquad a_i=0.
\]
Thus the expression $K_{v_a}-M_{m_a}$ means that all unordered pairs
of distinct directions are edges, except for these $m_a$ opposite
pairs.  In particular,
\[
|E(\Gamma_a)|=\binom{v_a}{2}-m_a.
\]
The graph $\Gamma_a$ is an ordinary simple graph and therefore has no
loops.  

In addition to the ordinary edge set $E(\Gamma_a)$, we introduce one formal loop symbol $\ell_u$ for every $u\in V(a)$, and put
\begin{equation}\label{eq:formal-edge-set}
 \widetilde E_a
 =E(\Gamma_a)\sqcup\{\ell_u:u\in V(a)\}.
\end{equation}
The ordinary edge $\{u,w\}$ corresponds to the point $a+u+w$, while $\ell_u$ corresponds to $a+2u$.  All these points lie in $S^n(r+1)$.  Moreover,
\begin{equation}\label{eq:formal-edge-count}
 |\widetilde E_a|
 =\binom{v_a}{2}-m_a+v_a
 =\binom{v_a}{2}+n.
\end{equation}

For $z\in\Omega(a)$, the difference $z-a$ has a unique expression
\begin{equation}\label{eq:beta-expansion}
 z-a=\sum_{u\in V(a)}\beta_a(z,u)u,
 \qquad
 \beta_a(z,u)\in\Z_{\ge0},
\end{equation}
in which at most one direction from each opposite pair has positive coefficient.  Since $a\preceq z$, there is no cancellation and
\begin{equation}\label{eq:beta-mass}
 \sum_{u\in V(a)}\beta_a(z,u)
 =\wtL(z)-\wtL(a)
 =r+2.
\end{equation}
Define
\begin{equation}\label{eq:RH-sets}
 \mathcal R_a(z)=\{u\in V(a):\beta_a(z,u)\ge1\},
 \qquad
 \mathcal H_a(z)=\{u\in V(a):\beta_a(z,u)\ge2\}.
\end{equation}
Then $\mathcal H_a(z)\subseteq\mathcal R_a(z)$, no opposite pair is contained in $\mathcal R_a(z)$, and
\begin{equation}\label{eq:RH-budget}
 |\mathcal R_a(z)|+|\mathcal H_a(z)|\le r+2.
\end{equation}
Indeed, every element of $\mathcal R_a(z)$ consumes one unit in \eqref{eq:beta-mass}, and every element of $\mathcal H_a(z)$ consumes at least one further unit.

For each $z\in\Omega(a)$, define the following actual subset of the formal edge set \eqref{eq:formal-edge-set}:
\begin{equation}\label{eq:Faz}
 \mathcal F_a(z)
 =\bigl\{\{u,w\}\in E(\Gamma_a):
          u,w\in\mathcal R_a(z)\bigr\}
  \sqcup
  \bigl\{\ell_u:u\in\mathcal H_a(z)\bigr\}.
\end{equation}
Thus $\mathcal F_a(z)$, rather than the ordered pair $(\mathcal R_a(z),\mathcal H_a(z))$, is the set of ordinary edges and loops assigned to $z$.

The next proposition is the exact local decomposition underlying the proof.

\begin{proposition}\label{prop:clique-decomp}
For every $a\in S^n(r-1)$,
\begin{equation}\label{eq:exact-formal-decomp}
 \widetilde E_a
 =\bigsqcup_{z\in\Omega(a)}\mathcal F_a(z).
\end{equation}
Equivalently, every ordinary edge of $\Gamma_a$ occurs in exactly one set $\mathcal F_a(z)$, and every formal loop $\ell_u$ occurs in exactly one such set.
\end{proposition}

\begin{proof}
Let $\{u,w\}\in E(\Gamma_a)$.  Then $u$ and $w$ are distinct and are not opposite, so
\[
 x=a+u+w\in S^n(r+1).
\]
The exact shell partition gives a unique $z\in\Omega$ with $x\preceq z$.  Since $a\preceq x\preceq z$, this center belongs to $\Omega(a)$, and $\beta_a(z,u),\beta_a(z,w)\ge1$.  Hence $\{u,w\}\in\mathcal F_a(z)$.  Conversely, if $\{u,w\}\in\mathcal F_a(z)$, then $a+u+w\preceq z$, so uniqueness in the shell partition shows that no other center can contain the same ordinary edge.

Now fix $u\in V(a)$.  The point
\[
 a+2u\in S^n(r+1)
\]
has a unique covering center $z$.  Again $a\preceq a+2u\preceq z$, so $z\in\Omega(a)$ and $\beta_a(z,u)\ge2$.  Therefore $\ell_u\in\mathcal F_a(z)$.  Conversely, $\ell_u\in\mathcal F_a(z)$ implies $a+2u\preceq z$, and exactness gives uniqueness.  This proves \eqref{eq:exact-formal-decomp}.
\end{proof}
In Figure \ref{fig:graph_model}, we provide an example in $\mathbb{Z}^2$ to illustrate Proposition \ref{prop:clique-decomp}. In this example $\Omega(a)=\{z,z'\}$, $\mathcal F_a(z)=\bigl\{\{-e_1,-e_2\}, \{\ell_{-e_1}\}, \{\ell_{-e_2}\}\bigr\}$ and $\mathcal F_a(z')=\bigl\{\{e_1,-e_2\}, \{\ell_{e_1}\}\bigr\}$.
\begin{figure}[h!]
	\centering
	\includegraphics[width=0.5\textwidth]{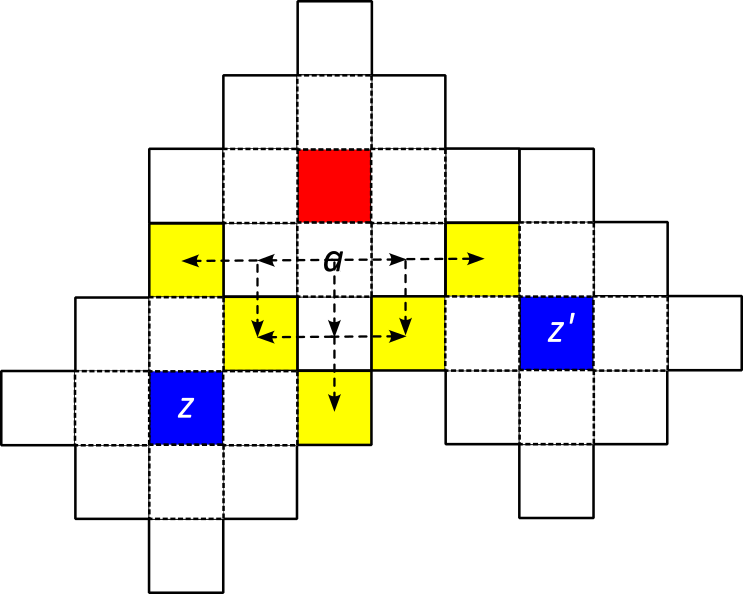}
	\centeredfigurecaption{An example of $\Omega(a)\subseteq S^2(5)$ and the corresponding partition of $\widetilde E_a$ \label{fig:graph_model}}
\end{figure}

\subsection{Incidence-rank lower bounds}\label{sec:descent-bounds}

We now establish the second step in the left branch of Fig.~\ref{fig:proof-framework}.  Incidence matrices convert the local decomposition into lower bounds for $|\Omega(a)|$, and hence into a lower bound for the sum of these quantities over $a\in S^n(r-1)$.

Let $N_a$ be the $0$-$1$ matrix whose rows are indexed by $V(a)$, whose columns are indexed by $z\in\Omega(a)$, and whose $z$-column is the incidence vector of $\mathcal R_a(z)$.  Let $\rho_u$ be the number of sets $\mathcal R_a(z)$ containing $u$.  Exactness on the ordinary edges gives
\begin{equation}\label{eq:NNt}
 N_aN_a^{\mathsf T}=J+D-P,
 \qquad
 D=\operatorname{diag}(\rho_u-1:u\in V(a)),
\end{equation}
where $P$ is the adjacency matrix of the missing matching in $\Gamma_a$.

The first consequence is a lower bound on the size of $\Omega(a)$. Note that the assumption on $r<n$ is essential: the example in Fig.~\ref{fig:graph_model} with $r=n=2$ fails to satisfy this lower bound.

\begin{lemma}\label{lem:weak}
If $r<n$, then for every $a\in S^n(r-1)$,
\begin{equation}\label{eq:weak-bound}
 |\Omega(a)|\ge n+1.
\end{equation}
\end{lemma}

\begin{proof}
We first show that $\rho_u\ge2$ for every $u\in V(a)$.  If $\rho_u=1$, then the unique set $\mathcal R_a(z)$ containing $u$ must contain every ordinary neighbor of $u$ in $\Gamma_a$.  By Proposition~\ref{prop:clique-decomp}, the same center is also the unique one for which $u\in\mathcal H_a(z)$.  Hence \eqref{eq:RH-budget} gives
\[
 |\mathcal R_a(z)|\le r+1,
\]
so that this set contains at most $r$ ordinary neighbors of $u$.  On the other hand,
\[
 \deg_{\Gamma_a}(u)
 \ge v_a-2
 =2n-|\supp(a)|-2
 \ge2n-r-1>r,
\]
because $|\supp(a)|\le r-1$ and $r<n$.  This contradiction proves $\rho_u\ge2$.

Thus every diagonal entry of $D$ in \eqref{eq:NNt} is at least one.  Moreover,
\[
 m_a=n-|\supp(a)|\ge n-r+1\ge2.
\]
By Lemma~\ref{lem:matching-matrix},
\[
 \operatorname{rank}(N_aN_a^{\mathsf T})
 \ge v_a-m_a+1=n+1.
\]
Therefore $N_a$ has at least $n+1$ columns, proving \eqref{eq:weak-bound}.
\end{proof}

The stronger bound follows when the compatibility graph is sufficiently large compared with the cardinality of the support of $a$.

\begin{lemma}\label{lem:strong}
If
\begin{equation}\label{eq:strong-pointwise-condition}
 2n-|\supp(a)|\ge2r+4,
\end{equation}
then
\begin{equation}\label{eq:strong-bound}
 |\Omega(a)|\ge2n-|\supp(a)|.
\end{equation}
In particular, \eqref{eq:strong-bound} holds for every $a\in S^n(r-1)$ whenever
\begin{equation}\label{eq:strong-uniform-condition}
 2n\ge3r+3.
\end{equation}
\end{lemma}

\begin{proof}
Suppose $\rho_u\le2$.  The unique center for which $u\in\mathcal H_a(z)$ contains at most $r$ ordinary neighbors of $u$, by \eqref{eq:RH-budget}.  A possible second set $\mathcal R_a(z')$ contains at most $r+1$ ordinary neighbors.  Hence
\[
 \deg_{\Gamma_a}(u)\le2r+1.
\]
But \eqref{eq:strong-pointwise-condition} gives
\[
 \deg_{\Gamma_a}(u)\ge v_a-2\ge2r+2,
\]
a contradiction.  Thus $\rho_u\ge3$ for every $u$.  Every diagonal entry of $D$ in \eqref{eq:NNt} is therefore at least two, so $D-P$ is positive definite by Lemma~\ref{lem:matching-matrix}.  Hence $N_aN_a^{\mathsf T}$ is positive definite and $N_a$ has row rank $v_a=2n-|\supp(a)|$.  This proves \eqref{eq:strong-bound}.

Finally, $|\supp(a)|\le r-1$, and $2n\ge3r+3$ implies
\[
 2n-|\supp(a)|
 \ge2n-r+1
 \ge2r+4. \qedhere
\]
\end{proof}

The full loop-and-edge decomposition also gives a direct bound that is useful when support sizes are treated separately.

\begin{lemma}\label{lem:loop-count}
For every $a\in S^n(r-1)$,
\begin{equation}\label{eq:loop-count}
 |\Omega(a)|
 \ge
 \left\lceil
 \frac{\binom{2n-|\supp(a)|}{2}+n}
      {\binom{r+2}{2}}
 \right\rceil.
\end{equation}
\end{lemma}

\begin{proof}
By \eqref{eq:formal-edge-count}, the formal edge set has
\[
 \binom{v_a}{2}+n
\]
elements.  If
\[
 s=|\mathcal R_a(z)|,
 \qquad
 h=|\mathcal H_a(z)|,
\]
then $\mathcal F_a(z)$ contains $\binom{s}{2}+h$ formal edges.  Since $s+h\le r+2$,
\[
 \binom{s}{2}+h
 \le\binom{s}{2}+r+2-s
 \le\binom{r+2}{2}.
\]
The last inequality follows from
\[
 \binom{r+2}{2}-\binom{s}{2}-(r+2-s)
 =\frac{(r+2-s)(r+s-1)}2\ge0.
\]
Together with the exact partition \eqref{eq:exact-formal-decomp} and the inequalities obtained above, we get \eqref{eq:loop-count}.
\end{proof}

\subsection{Coefficient gaps and the upper bound}\label{sec:coefficient}

We continue to work on the right branch of Fig.~\ref{fig:proof-framework} in this subsection.

For $z\in\mathbb{Z}^n$, define
\begin{equation}\label{eq:Pz}
 P_z(q)=\prod_{i=1}^n(1+q+\cdots+q^{|z_i|})
       =\sum_{j=0}^{2r+1}c_j(z)q^j.
\end{equation}
The coefficient $c_j(z)$ is the number of points $x\in S^n(j)$ satisfying $x\preceq z$: the absolute value of each coordinate may be chosen between zero and $|z_i|$, and its sign is then uniquely determined.

By Lemma~\ref{lem:unimodal}, the coefficient sequence of $P_z(q)$ is symmetric unimodal.  Since its total degree is $2r+1$,
\begin{equation}\label{eq:central-coeff-basic}
 c_{r-1}(z)\le c_r(z)=c_{r+1}(z).
\end{equation}
Double counting pairs $(a,z)$ with $a\in S^n(r-1)$ and $a\preceq z$, and then using the exact shell partition, gives
\begin{equation}\label{eq:global-upper}
 \sum_{a\in S^n(r-1)}|\Omega(a)|
 =\sum_{z\in\Omega}c_{r-1}(z)
 \le\sum_{z\in\Omega}c_{r+1}(z)
 =L_{r+1}.
\end{equation}

For $j\ge0$, write
\[
 F_j(q)=1+q+\cdots+q^j.
\]
The next identity is used to make the inequality in \eqref{eq:central-coeff-basic} quantitative.

\begin{lemma}\label{lem:polarization}
If $A\ge B\ge2$, then
\begin{equation}\label{eq:polarization}
 F_A(q)F_B(q)-F_{A+1}(q)F_{B-1}(q)
 =q^B F_{A-B}(q).
\end{equation}
\end{lemma}

\begin{proof}
Use $F_j(q)=(1-q^{j+1})/(1-q)$.  After multiplication by $(1-q)^2$, the numerator on the left is
\[
 (1-q^{A+1})(1-q^{B+1})
 -(1-q^{A+2})(1-q^B),
\]
which simplifies to
\[
 q^B(1-q)(1-q^{A-B+1}).
\]
Dividing by $(1-q)^2$ proves the identity.
\end{proof}

For $z\in S^n(2r+1)$, put
\[
 s(z)=|\supp(z)|
\]
and define
\begin{equation}\label{eq:central-gap}
 g(z)=c_{r+1}(z)-c_{r-1}(z)
     =c_r(z)-c_{r-1}(z).
\end{equation}
The following estimate is the main quantitative refinement of symmetric unimodality.

\begin{lemma}\label{lem:central-gap}
For every $z\in S^n(2r+1)$,
\begin{equation}\label{eq:central-gap-sharp}
 g(z)\ge
 \frac{2r+3-s(z)}{r+2}
 \binom{s(z)}{r+1}.
\end{equation}
Consequently,
\begin{equation}\label{eq:central-gap-weak}
 g(z)\ge
 \frac2{r+2}\binom{s(z)}{r+1}.
\end{equation}
\end{lemma}

\begin{proof}
List the positive absolute coordinate values of $z$ as
\[
 \alpha_1,\ldots,\alpha_{s(z)},
 \qquad
 \sum_i\alpha_i=2r+1.
\]
Suppose that two coordinates of $z$, say the $i$th and $j$th
coordinates, satisfy
\[
|z_i|=A\ge B=|z_j|\ge2.
\]
Let $z'\in S^n(2r+1)$ be obtained from $z$ by transferring one unit
from the $j$th absolute coordinate to the $i$th absolute coordinate.
More precisely, put
\[
z_i'=z_i+\sgn(z_i),\qquad
z_j'=z_j-\sgn(z_j),
\]
and let $z_\ell'=z_\ell$ for $\ell\notin\{i,j\}$.  Since $B\ge2$,
the $j$th coordinate of $z'$ is still nonzero.  In particular,
\[
\wtL(z')=\wtL(z)=2r+1
\qquad\text{and}\qquad
|\supp(z')|=|\supp(z)|.
\]

By Lemma~\ref{lem:polarization},
\[
	P_z(q)-P_{z'}(q)
	\bigl(F_A(q)F_B(q)
	-F_{A+1}(q)F_{B-1}(q)\bigr)
	\prod_{\ell\notin\{i,j\}}F_{|z_\ell|}(q)
	=q^B H(q),
\]
where
\[
H(q)
=
F_{A-B}(q)
\prod_{\ell\notin\{i,j\}}F_{|z_\ell|}(q).
\]
The total degree of $H(q)$ is
\[
(A-B)+(2r+1-A-B)=2(r-B)+1.
\]
Write
\[
H(q)=\sum_{t\ge0}h_tq^t,
\]
where we use the convention $h_t=0$ when $t<0$.  By
Lemma~\ref{lem:unimodal}, the coefficient sequence of $H(q)$ is
symmetric unimodal.  Hence
\[
h_{r-B}\ge h_{r-B-1}.
\]
Here $B\le r$, since $A\ge B$ and $A+B\le2r+1$.

Using the definition of the central gap, we now obtain
\[
\begin{aligned}
	g(z)-g(z')
	&=
	\bigl([q^r]P_z(q)-[q^{r-1}]P_z(q)\bigr)-
	\bigl([q^r]P_{z'}(q)-[q^{r-1}]P_{z'}(q)\bigr)\\
	&=[q^r]\bigl(q^BH(q)\bigr)
	-[q^{r-1}]\bigl(q^BH(q)\bigr)\\
	&=h_{r-B}-h_{r-B-1}\\
	&\ge0,
\end{aligned}
\]
where $[q^t]Q(q)$ denotes the coefficient of $q^t$ in a polynomial
$Q(q)$.

Thus transferring one unit from a smaller capacity of size at least
two to a capacity at least as large cannot increase the central gap.
Repeating this operation shows that, among all positive capacity
vectors of length $s=s(z)$ and total sum $2r+1$, the central gap is
minimized at
\[
(2r+2-s,1,\ldots,1).
\]
For this vector,
\[
 P(q)=F_{2r+2-s}(q)(1+q)^{s-1}.
\]
Moreover,
\[
	 g_{\min}
	 =\binom{s-1}{r}-\binom{s-1}{r+2}
	 =\binom{s}{r+1}-\binom{s}{r+2}
	 =\frac{2r+3-s}{r+2}\binom{s}{r+1}.
\]
This proves \eqref{eq:central-gap-sharp}.  Since $s\le2r+1$, \eqref{eq:central-gap-weak} follows directly.
\end{proof}

There are
\begin{equation}\label{eq:unit-coordinate-count}
 2^{r+1}\binom n{r+1}
\end{equation}
unit-coordinate points in $S^n(r+1)$ when $r<n$.  A center $z$ dominates exactly $\binom{s(z)}{r+1}$ of them.  Hence the exact shell partition gives
\begin{equation}\label{eq:unit-coordinate-identity}
 \sum_{z\in\Omega}\binom{s(z)}{r+1}
 =2^{r+1}\binom n{r+1}.
\end{equation}
Combining \eqref{eq:central-gap-weak} and \eqref{eq:unit-coordinate-identity} yields
\begin{equation}\label{eq:total-gap}
 \sum_{z\in\Omega}g(z)
 \ge
 \frac{2^{r+2}}{r+2}\binom n{r+1}.
\end{equation}
Consequently,
\begin{equation}\label{eq:refined-global}
 \sum_{a\in S^n(r-1)}|\Omega(a)|
 \le
 L_{r+1}-\frac{2^{r+2}}{r+2}\binom n{r+1}.
\end{equation}
Indeed,
\[
 \sum_{a\in S^n(r-1)}|\Omega(a)|
 =\sum_{z\in\Omega}c_{r-1}(z)
 =\sum_{z\in\Omega}\bigl(c_{r+1}(z)-g(z)\bigr),
\]
and $\sum_z c_{r+1}(z)=L_{r+1}$.

Define
\begin{equation}\label{eq:Phi}
 \Phi_{n,r}
 =rL_r-(r-1)L_{r-1}-L_{r+1}
  +\frac{2^{r+2}}{r+2}\binom n{r+1}.
\end{equation}
The lower and upper incidence estimates give the following exact criterion.

\begin{theorem}\label{thm:exact-criterion}
Let $n\ge3$, and let $r\ge2$ be an integer satisfying
\[
 r\le\left\lfloor\frac{2n-3}{3}\right\rfloor.
\]
If $\Phi_{n,r}>0$, then no complete covering of $B^n(r)$ exists.
\end{theorem}

\begin{proof}
The upper bound on $r$ implies $2n\ge3r+3$ and, in particular, $r<n$.  By Lemma~\ref{lem:strong} and \eqref{eq:support-identity},
\[
	 \sum_{a\in S^n(r-1)}|\Omega(a)|
	 \ge
	 \sum_{a\in S^n(r-1)}
	 \bigl(2n-|\supp(a)|\bigr)
	 =2nL_{r-1}-K_{r-1}
	 =rL_r-(r-1)L_{r-1}.
\]
The refined upper bound \eqref{eq:refined-global} therefore forces $\Phi_{n,r}\le0$.  Hence $\Phi_{n,r}>0$ is impossible.
\end{proof}

For later use, we record a closed form and a monotonicity property.

\begin{proposition}\label{prop:Phi-closed}
For $r<n$,
\begin{equation}\label{eq:Phi-closed}
	 \Phi_{n,r}
	 = \sum_{j=2}^{r-1}	 2^j\binom{n}{j} \frac{(j-1)(r-1)}{r-j+1}
	 \binom{r-2}{j-1} -\frac{2^{r+1}r}{r+2}\binom n{r+1}.
\end{equation}
For fixed $r$, the normalized quantity
\[
 \frac{\Phi_{n,r}}{\binom n{r+1}}
\]
is strictly decreasing for $n\ge r+1$.
\end{proposition}

\begin{proof}
Substitute \eqref{eq:shell-explicit} into \eqref{eq:Phi} and collect terms according to the support size $j$.  For $2\le j\le r-1$,
\[
	 r\binom{r-1}{j-1}
	 -(r-1)\binom{r-2}{j-1}
	 -\binom r{j-1}\
	=
	 \frac{(j-1)(r-1)}{r-j+1}
	 \binom{r-2}{j-1}.
\]
The support-one contribution is zero, the support-$r$ contributions cancel, and the support-$(r+1)$ contribution together with the correction term in \eqref{eq:Phi} is
\[
 -\frac{2^{r+1}r}{r+2}\binom n{r+1}.
\]
This proves \eqref{eq:Phi-closed}.

After division by $\binom n{r+1}$, the final term is independent of $n$, whereas every ratio
\[
 \frac{\binom nj}{\binom n{r+1}},
 \qquad 2\le j\le r-1,
\]
is strictly decreasing in $n$.  All its coefficients in \eqref{eq:Phi-closed} are positive, proving the assertion.
\end{proof}

\subsection{Finite comparison of the two bounds}\label{sec:uniform-proof}

In this part, we combine the lower and upper estimates, use the Lee-shell recurrence to exclude the region $2n<r^2$, and then apply the exact quantity $\Phi_{n,r}$ in the remaining finite-parameter range to finish the proof of Theorem~\ref{thm:main}.

We first settle the range in which $r^2>2n$.  The first lemma treats the part where the strong rank bound applies.

\begin{lemma}\label{lem:square-strong}
Let $n\ge9$, and let $r$ be an integer satisfying
\[
 \max\left\{
 4,
 \left\lfloor\sqrt{2n}\right\rfloor+1
 \right\}
 \le r\le
 \left\lfloor\frac{2n-3}{3}\right\rfloor.
\]
Then no complete covering of $B^n(r)$ exists.
\end{lemma}

\begin{proof}
Suppose that $\Omega$ is a complete covering.  By Lemma~\ref{lem:strong}, \eqref{eq:support-identity}, and \eqref{eq:global-upper},
\begin{equation}\label{eq:strong-necessary-main}
 L_{r+1}\ge rL_r-(r-1)L_{r-1}.
\end{equation}
Put
\[
 \delta=r^2-2n.
\]
Then $\delta\ge1$.  The recurrence \eqref{eq:recurrence}, first with $t=r-1$, gives
\[
 \frac{L_r}{L_{r-1}}
 =\frac{2n}{r}+\frac{r-2}{r}\frac{L_{r-2}}{L_{r-1}}
 >r-\frac{\delta}{r}.
\]
Using the same recurrence with $t=r$, we obtain
\[
 \begin{aligned}
 &\frac{r(r+1)}{L_{r-1}}
 \bigl(rL_r-(r-1)L_{r-1}-L_{r+1}\bigr)\\
 &\qquad>
 -\delta^2+\delta r^2-\delta r-r^2+2r.
 \end{aligned}
\]
The assumption $2n\ge3r+3$ implies
\[
 1\le\delta\le r^2-3r-3.
\]
The expression
\[
 -\delta^2+\delta r^2-\delta r-r^2+2r
\]
is a concave quadratic polynomial in $\delta$.  At the two endpoints of the displayed interval it has the values
\[
 r-1
\]
and
\[
 2r^3-4r^2-13r-9,
\]
respectively.  Both values are positive for $r\ge4$.  Hence
\[
 rL_r-(r-1)L_{r-1}>L_{r+1},
\]
contradicting \eqref{eq:strong-necessary-main}.
\end{proof}

The complementary part of the range $r^2>2n$ follows from the weaker rank bound.

\begin{lemma}\label{lem:square-weak}
Let $n\ge10$, and let $r$ be an integer satisfying
\[
 \max\left\{
 \left\lfloor\sqrt{2n}\right\rfloor+1,
 \left\lceil\frac{2n-2}{3}\right\rceil
 \right\}
 \le r<n.
\]
Then no complete covering of $B^n(r)$ exists.
\end{lemma}

\begin{proof}
If a complete covering existed, Lemma~\ref{lem:weak} and \eqref{eq:global-upper} would give
\begin{equation}\label{eq:weak-necessary-main}
 L_{r+1}\ge(n+1)L_{r-1}.
\end{equation}
The recurrence \eqref{eq:recurrence} with $t=r-2$ gives
\[
 \frac{L_{r-2}}{L_{r-1}}\le\frac{r-1}{2n}.
\]
Using it with $t=r-1$, we obtain
\[
 \frac{L_r}{L_{r-1}}
 \le
 \frac{2n}{r}+
 \frac{(r-2)(r-1)}{2nr}.
\]
A further use of the recurrence, now with $t=r$, gives
\begin{equation}\label{eq:ratio-upper-main}
 \frac{L_{r+1}}{L_{r-1}}
 \le
 \frac{4n^2+2(r-1)^2}{r(r+1)}.
\end{equation}
For fixed $n$ and real $1\le x<n$, the function
\[
 x\longmapsto
 \frac{4n^2+2(x-1)^2}{x(x+1)}
\]
is strictly decreasing, because its derivative is
\[
 \frac{2(3x^2-2x-1-4n^2x-2n^2)}{x^2(x+1)^2}<0.
\]
The inequality $2n\le3r+2$ gives $r\ge2(n-1)/3$.  Substituting $x=2(n-1)/3$ into the difference between $n+1$ and the right-hand side of \eqref{eq:ratio-upper-main}, and multiplying by the positive denominator, gives
\[
 \frac{2}{9}\bigl(2n^3-21n^2+18n-26\bigr).
\]
This expression is positive for every $n\ge10$: at $n=10$ it equals $54$, and its derivative is positive for $n\ge10$.  Hence
\[
 L_{r+1}<(n+1)L_{r-1},
\]
contradicting \eqref{eq:weak-necessary-main}.
\end{proof}

The two dimensions not covered uniformly by Lemma~\ref{lem:square-weak} are handled separately.  Isolating them here keeps the proof of Theorem~\ref{thm:main} from being obscured by short numerical checks.

\begin{lemma}\label{lem:eight-nine}
Let $n\in\{8,9\}$ and $2\le r<n$.  If $r^2>2n$, then no complete covering of $B^n(r)$ exists.
\end{lemma}

\begin{proof}
First let $n=9$.  The value $r=5$ satisfies $2n=3r+3$ and is excluded by Lemma~\ref{lem:square-strong}.  For $r=6,7,8$, Lemma~\ref{lem:weak} would imply
\[
 L_{r+1}\ge10L_{r-1}.
\]
However, the explicit shell formula \eqref{eq:shell-explicit} gives
\[
 \begin{array}{c|ccc}
 r&6&7&8\\ \hline
 10L_{r-1}-L_{r+1}&18594&157266&622114.
 \end{array}
\]
Each entry is positive, giving a contradiction.

Now let $n=8$.  Since $r<n$ and $r^2>16$, we have $r\in\{5,6,7\}$.  For $r=6,7$, the weak bound would give $L_{r+1}\ge9L_{r-1}$, whereas
\[
 9L_5-L_7=16352,
 \qquad
 9L_6-L_8=85888.
\]
It remains to consider $r=5$.  The shell $S^8(4)$ contains, respectively,
\[
 16,\qquad336,\qquad1344,\qquad1120
\]
points of support sizes $1,2,3,4$.  For support sizes $1$ and $2$, the pointwise hypothesis of Lemma~\ref{lem:strong} is satisfied, and the corresponding lower bounds for $|\Omega(a)|$ are $15$ and $14$.  For support sizes $3$ and $4$, Lemma~\ref{lem:weak} gives the lower bound $9$.  It follows that
\[
 \begin{aligned}
 \sum_{a\in S^8(4)}|\Omega(a)|
 &\ge16\cdot15+336\cdot14
      +1344\cdot9+1120\cdot9\\
 &=27120.
 \end{aligned}
\]
On the other hand, \eqref{eq:global-upper} gives an upper bound of
\[
 L_6=27008,
\]
which is a contradiction.
\end{proof}

We can now prove the first main theorem.

\begin{proof}[Proof of Theorem~\ref{thm:main}]
Suppose, for a contradiction, that $\Omega$ is a complete covering satisfying \eqref{eq:main-condition}.  If $r\ge n$, then Corollary~\ref{cor:spacapan-Zn} gives a contradiction.  We may therefore assume
\begin{equation}\label{eq:r-less-n-main}
 r<n.
\end{equation}

Assume first that $r^2>2n$.  If $n\ge10$, then either $2n\ge3r+3$ or $2n\le3r+2$, so Lemma~\ref{lem:square-strong} or Lemma~\ref{lem:square-weak} applies.  If $n\in\{8,9\}$, the conclusion follows from Lemma~\ref{lem:eight-nine}.  Thus no complete covering exists in this range.

It remains to consider $r^2\leq 2n$.  
The hypothesis \eqref{eq:main-condition} gives
\begin{equation}\label{eq:p-range-main}
 0\le 2n-r^2
 \le\left\lfloor\frac{25r}{34}\right\rfloor-3
 \le\frac{25r}{34}-3.
\end{equation}
In particular, $r\ge5$.  Moreover, as $2n\ge3r+3$,
Theorem~\ref{thm:exact-criterion} is applicable.  We shall prove that $\Phi_{n,r}$ which is defined in \eqref{eq:Phi-closed} is positive.

The recurrence \eqref{eq:recurrence} gives
\[
 \frac{L_r}{L_{r-1}}
 >\frac{2n}{r}=r+\frac pr.
\]
Using the recurrence once more, we obtain
\begin{equation}\label{eq:shell-difference-lower}
 \frac{rL_r-(r-1)L_{r-1}-L_{r+1}}{L_{r-1}}
 >-
 \frac{r-2+p(r-1)+p^2/r}{r+1}.
\end{equation}

To estimate the positive correction term in \eqref{eq:Phi}, retain the contribution
\[
 2^{r-1}\binom n{r-1}
\]
of the points of support $r-1$ in $L_{r-1}$.  Writing the remaining support sizes as $r-1-\ell$, the ratio of the corresponding term in \eqref{eq:shell-explicit} to this top-support contribution is at most
\[
 \frac1{\ell!}
 \left(
  \frac{(r-1)(r-2)}{2(n-r+2)}
 \right)^\ell.
\]
Summing them gives
\begin{equation}\label{eq:Ltop-main}
 L_{r-1}
 \le
 2^{r-1}\binom n{r-1}
 \exp\left(
  \frac{(r-1)(r-2)}{2(n-r+2)}
 \right).
\end{equation}
Consequently,
\begin{equation}\label{eq:correction-lower-main}
 \begin{aligned}
 &\frac{2^{r+2}}{(r+2)L_{r-1}}\binom n{r+1}\\
 &\quad\ge
 \frac{8(n-r)(n-r+1)}{r(r+1)(r+2)}
 \exp\left(
  -\frac{(r-1)(r-2)}{2(n-r+2)}
 \right).
 \end{aligned}
\end{equation}

Apply Lemma~\ref{lem:comparison-25-34} from the appendix with $a=n-r$.  Its hypotheses follow from \eqref{eq:p-range-main}.  Combining \eqref{eq:shell-difference-lower} and \eqref{eq:correction-lower-main}, we obtain
\[
 \frac{\Phi_{n,r}}{L_{r-1}}>0.
\]
This contradicts Theorem~\ref{thm:exact-criterion} and completes the proof.
\end{proof}

\subsection{Asymptotic comparison of the two bounds}\label{sec:asymptotic-proof}

In this subsection, we evaluate the two sides of the necessary inequality when $2n=r^2+O(r)$, or equivalently evaluate $\Phi_{n,r}$ in this range, and then prove Theorem~\ref{thm:asymptotic-intro}.

We next determine the asymptotic size of $L_{r-1}$ under the condition that $n\sim r^2/2$.  The proof is included in detail because the factor $\mathrm e$ is responsible for the coefficient $2/\mathrm e$ in Theorem~\ref{thm:asymptotic-intro}.

\begin{lemma}\label{lem:top-support}
Suppose that $r\to\infty$, that $n=n(r)$ is integer-valued, and that
\[
 2n=r^2+O(r).
\]
Then
\begin{equation}\label{eq:top-support-e}
 \frac{L_{r-1}}{2^{r-1}\binom n{r-1}}
 \longrightarrow\mathrm e.
\end{equation}
\end{lemma}

\begin{proof}
We expand $L_{r-1}$ according to the support size.  A point of $S^n(r-1)$ with exactly $j$ nonzero coordinates is obtained by choosing the $j$ positions, choosing their signs, and choosing a composition of $r-1$ into $j$ positive parts.  Thus, for all sufficiently large $r$,
\begin{equation}\label{eq:Lrminus1-support}
 L_{r-1}
 =\sum_{j=1}^{r-1}
 2^j\binom nj\binom{r-2}{j-1}.
\end{equation}
The condition $2n=r^2+O(r)$ implies $n\sim r^2/2$, so $n\ge r-1$ for all sufficiently large $r$.

The term with the largest possible support, $j=r-1$, is
\[
 T_0=2^{r-1}\binom n{r-1}.
\]
It counts the points whose nonzero coordinates all have absolute value one.  Write
\[
 j=r-1-\ell,
 \qquad 0\le\ell\le r-2.
\]
Since
\[
 \binom{r-2}{j-1}=\binom{r-2}{\ell},
\]
the ratio of the corresponding summand in \eqref{eq:Lrminus1-support} to $T_0$ is
\begin{equation}\label{eq:ratio-ell}
 R_{r,\ell}
 =2^{-\ell}
 \frac{\binom n{r-1-\ell}}{\binom n{r-1}}
 \binom{r-2}{\ell}.
\end{equation}
Hence
\begin{equation}\label{eq:normalized-L-sum}
 \frac{L_{r-1}}{T_0}
 =\sum_{\ell=0}^{r-2}R_{r,\ell}.
\end{equation}

For $\ell\ge1$,
\[
 \frac{\binom n{r-1-\ell}}{\binom n{r-1}}
 =\prod_{s=0}^{\ell-1}
 \frac{r-1-s}{n-r+2+s},
\]
and
\[
 \binom{r-2}{\ell}
 =\frac1{\ell!}
 \prod_{s=0}^{\ell-1}(r-2-s).
\]
Substitution into \eqref{eq:ratio-ell} gives the exact product formula
\begin{equation}\label{eq:ratio-ell-product}
 R_{r,\ell}
 =\frac1{\ell!}
 \prod_{s=0}^{\ell-1}
 \frac{(r-1-s)(r-2-s)}{2(n-r+2+s)}.
\end{equation}
The empty product is interpreted as one when $\ell=0$.

Fix $\ell\ge0$.  For every fixed $s\in\{0,\ldots,\ell-1\}$,
\[
 (r-1-s)(r-2-s)=r^2+O_\ell(r),
\]
whereas
\[
 2(n-r+2+s)=r^2+O_\ell(r).
\]
Every factor in \eqref{eq:ratio-ell-product} therefore tends to one, and hence
\begin{equation}\label{eq:ratio-pointwise}
 R_{r,\ell}\longrightarrow\frac1{\ell!}
 \qquad\text{for every fixed }\ell.
\end{equation}


On the other hand, the factors in \eqref{eq:ratio-ell-product} satisfy
\[
 (r-1-s)(r-2-s)\le(r-1)(r-2)
\]
and
\[
 n-r+2+s\ge n-r+2.
\]
It follows that
\begin{equation}\label{eq:ratio-majorant}
 0\le R_{r,\ell}
 \le\frac{\theta_r^\ell}{\ell!},
 \qquad
 \theta_r=
 \frac{(r-1)(r-2)}{2(n-r+2)}.
\end{equation}
Furthermore,
\[
 \theta_r
 =\frac{r^2-3r+2}{2n-2r+4}
 \longrightarrow1.
\]
Thus $\theta_r\le2$ for all sufficiently large $r$.  Set $R_{r,\ell}=0$ when $\ell>r-2$.  Then
\[
 0\le R_{r,\ell}\le\frac{2^\ell}{\ell!}
\]
for every $\ell\ge0$ and all sufficiently large $r$.

For completeness, we finish with an explicit finite-part and tail argument.  Let $\eta>0$.  Choose $M$ so large that
\[
 \sum_{\ell>M}\frac{2^\ell}{\ell!}<\eta
 \qquad\text{and}\qquad
 \sum_{\ell>M}\frac1{\ell!}<\eta.
\]
By \eqref{eq:ratio-pointwise}, for all sufficiently large $r$,
\[
 \left|
 \sum_{\ell=0}^{M}R_{r,\ell}
 -\sum_{\ell=0}^{M}\frac1{\ell!}
 \right|<\eta.
\]
The tail estimate \eqref{eq:ratio-majorant} gives
\[
 0\le\sum_{\ell>M}R_{r,\ell}<\eta.
\]
Consequently,
\[
 \left|
 \sum_{\ell=0}^{\infty}R_{r,\ell}
 -\sum_{\ell=0}^{\infty}\frac1{\ell!}
 \right|<3\eta
\]
for all sufficiently large $r$.  Since $\eta$ is arbitrary, \eqref{eq:normalized-L-sum} yields
\[
 \frac{L_{r-1}}{T_0}
 \longrightarrow
 \sum_{\ell=0}^{\infty}\frac1{\ell!}
 =\mathrm e,
\]
which is \eqref{eq:top-support-e}.
\end{proof}

The preceding lemma allows us to evaluate the exact criterion on the scale $2n=r^2+O(r)$.

\begin{proposition}\label{prop:Phi-asymptotic}
Let $c<1$ be fixed and suppose that
\[
 2n=r^2+cr+O(1).
\]
Then
\begin{equation}\label{eq:Phi-asymptotic}
 \frac{\Phi_{n,r}}{L_{r-1}}
 =\left(\frac2{\mathrm e}-c\right)r+o(r).
\end{equation}
\end{proposition}

\begin{proof}
The recurrence \eqref{eq:recurrence} with $t=r-2$ gives
\[
 \frac{L_{r-2}}{L_{r-1}}=O(r^{-1}).
\]
Using it with $t=r-1$, we obtain
\[
 \frac{L_r}{L_{r-1}}=r+c+o(1).
\]
A further use of the recurrence, now with $t=r$, gives
\begin{equation}\label{eq:asymptotic-shell-difference}
 \frac{rL_r-(r-1)L_{r-1}-L_{r+1}}{L_{r-1}}
 =-cr+o(r).
\end{equation}

For the positive correction term in \eqref{eq:Phi}, a direct binomial-ratio calculation gives
\[
 \frac{\displaystyle\frac{2^{r+2}}{r+2}\binom n{r+1}}
      {2^{r-1}\binom n{r-1}}
 =\frac{8(n-r)(n-r+1)}{r(r+1)(r+2)}
 =2r+O(1).
\]
By Lemma~\ref{lem:top-support},
\begin{equation}\label{eq:asymptotic-correction}
 \frac{\displaystyle\frac{2^{r+2}}{r+2}\binom n{r+1}}
      {L_{r-1}}
 =\frac2{\mathrm e}r+o(r).
\end{equation}
Adding \eqref{eq:asymptotic-shell-difference} and \eqref{eq:asymptotic-correction} proves \eqref{eq:Phi-asymptotic}.
\end{proof}

\begin{proof}[Proof of Theorem~\ref{thm:asymptotic-intro}]
Fix $\eps>0$ and take $r$ sufficiently large.  If $r\ge n$, nonexistence follows from Corollary~\ref{cor:spacapan-Zn}.  We may therefore assume $r<n$.

If $2n<r^2$, Lemmas~\ref{lem:square-strong} and~\ref{lem:square-weak} apply once $r$ is sufficiently large, according as $2n\ge3r+3$ or $2n\le3r+2$.  It remains to consider
\[
 r^2\le2n
 \le r^2+\left(\frac2{\mathrm e}-\eps\right)r.
\]
By Proposition~\ref{prop:Phi-closed}, the normalized quantity
\[
 \frac{\Phi_{n,r}}{\binom n{r+1}}
\]
is decreasing in $n$.  Hence the least favorable case in this interval occurs at the largest admissible integer $n$, for which
\[
 2n=r^2+\left(\frac2{\mathrm e}-\eps\right)r+O(1).
\]
Proposition~\ref{prop:Phi-asymptotic} gives
\[
 \frac{\Phi_{n,r}}{L_{r-1}}
 =\eps r+o(r)>0
\]
for all sufficiently large $r$.  Moreover, $2n\ge3r+3$ in this range once $r$ is large.  Theorem~\ref{thm:exact-criterion} now gives the required contradiction.
\end{proof}

\begin{remark}\label{rem:two-over-e}
The coefficient $2/\mathrm e$ comes from two sharp features of the present argument.  First, unit-coordinate centers attain equality in \eqref{eq:central-gap-weak}.  Second, the factor $\mathrm e$ in \eqref{eq:top-support-e} is the limiting ratio between the whole shell $S^n(r-1)$ and its top-support contribution.  An improvement of the linear correction beyond $2/\mathrm e$ must therefore use information coupling different centers, rather than a stronger pointwise estimate depending only on the support size of one center.
\end{remark}

\section{Conclusion and Further Remarks}\label{sec:conclusion}

A complete covering records exactly the first-shell condition forced around each codeword of a perfect Lee code in $\Z^n$.  The shadow formulation shows that its centers form a constant-Lee-weight code attaining equality in a weighted shadow-packing inequality.  In dimension two this local condition is rigid and recovers a lattice tiling, while the explicit configuration for $(n,r)=(3,2)$ shows that complete coverings need not extend to global tilings in higher dimension.

The local decomposition associated with $\Gamma_a$ retains both pairs of distinct outward directions and repeated directions.  Its incidence-rank consequences, combined with the coefficient gap $g(z)$, prove Theorems~\ref{thm:main} and~\ref{thm:asymptotic-intro}.  Since every perfect Lee code induces a complete covering, these local nonexistence results give new cases of the Golomb--Welch conjecture.  In the opposite asymptotic regime, the near-optimal ternary constant-weight codes in \cite{LiuShangguan2025} produce almost complete coverings for every fixed radius.

The transfer result in Proposition~\ref{prop:cycle-lifting}\textup{(ii)} also applies Theorems~\ref{thm:main} and~\ref{thm:asymptotic-intro} to sufficiently long Cartesian products of cycles.  More precisely, if
\[
 G=C_{\ell_1}\times \cdots\times C_{\ell_n},
 \qquad
 \ell_i\ge4r+3\quad(1\le i\le n),
\]
then no radius-$r$ ball in $G$ admits an optimal local structure in either parameter range of the two main theorems.  This complements Theorem~\ref{thm:spacapan}: \v{S}pacapan allows the shorter condition $\ell_i\ge2r+1$ in the range $r\ge n$, whereas the transfer proposition applies the stronger restrictions on $n$ and $r$ obtained here at the cost of longer cycles.

The main remaining questions are whether exact complete coverings exist in further dimensions and whether the local nonexistence range can be extended below the square-root scale.  The local decompositions associated with the graphs $\Gamma_a$ suggest that stronger restrictions may come from compatibility among the decompositions attached to different points $a\in S^n(r-1)$, rather than from separate rank estimates at individual points.

\appendices

\section{The Technical Comparison for the Constant 25/34}\label{app:comparison}

The following elementary estimate is separated from the main argument because its proof is computational and independent of the combinatorial ideas.

\begin{lemma}\label{lem:comparison-25-34}
Let $r\ge5$, let $0\le p\le25r/34-3$, and put
\[
 a=\frac{r^2-2r+p}{2}.
\]
Then
\begin{equation}\label{eq:comparison-25-34}
 \begin{aligned}
 &\frac{8a(a+1)}{r(r+1)(r+2)}
 \exp\left(-\frac{(r-1)(r-2)}{2(a+2)}\right)\\
 &\qquad>
 \frac{r-2+p(r-1)+p^2/r}{r+1}.
 \end{aligned}
\end{equation}
\end{lemma}

\begin{proof}
Denote the left-hand and right-hand sides of \eqref{eq:comparison-25-34} by $G_r(p)$ and $T_r(p)$, respectively.  Since $da/dp=1/2$, differentiation gives
\[
 \begin{aligned}
 G_r'(p)
 ={}&\frac{4\exp\bigl(-(r-1)(r-2)/(2(a+2))\bigr)}
          {r(r+1)(r+2)}\\
 &\times\left(
  2a+1+
  \frac{(r-1)(r-2)a(a+1)}{2(a+2)^2}
 \right).
 \end{aligned}
\]
The expression in parentheses is increasing in $a$, because
\[
 \frac{d}{da}\left(\frac{a(a+1)}{(a+2)^2}\right)
 =\frac{3a+2}{(a+2)^3}>0.
\]
At the right endpoint $p_*=25r/34-3$, we have
\[
 a_* =\frac{34r^2-43r-102}{68}.
\]
The exponential factor is less than one.  Direct simplification gives
\[
 \begin{aligned}
 &r(r+2)(r-1)-4\left(
 2a_*+1+
 \frac{(r-1)(r-2)a_*(a_*+1)}{2(a_*+2)^2}
 \right)\\
 &\quad=
 \frac{1}{17(34r^2-43r+34)^2}
 \Bigl(
 19652(r-5)^7+539852(r-5)^6\\
 &\hspace{15mm}
 +6375561(r-5)^5+42382768(r-5)^4
 +173212044(r-5)^3\\
 &\hspace{15mm}
 +439788886(r-5)^2+647147455(r-5)+426965142
 \Bigr),
 \end{aligned}
\]
which is positive for $r\ge5$.  Therefore
\[
 G_r'(p)<\frac{r-1}{r+1}.
\]
On the other hand,
\[
 T_r'(p)=\frac{r-1+2p/r}{r+1}
 \ge\frac{r-1}{r+1}.
\]
Thus $G_r(p)-T_r(p)$ is strictly decreasing in $p$, and it remains to verify positivity at $p=p_*$.

At this endpoint,
\[
 \frac{(r-1)(r-2)}{2(a_*+2)}
 =1-\frac{59r-34}{34r^2-43r+34}.
\]
We use the elementary estimate $\mathrm e<68/25$.  Indeed,
\[
 \sum_{j=0}^4\frac1{j!}=\frac{65}{24},
 \qquad
 \sum_{j=5}^{\infty}\frac1{j!}
 <\frac1{120}\sum_{j\ge0}\frac1{6^j}
 =\frac1{100}.
\]
Hence $\mathrm e<1631/600<68/25$.  Since $\mathrm e^x>1+x$ for $x>0$, substitution into $G_r(p_*)$ gives
\[
 \begin{aligned}
 G_r(p_*)-T_r(p_*)
 >{}&
 \frac{1}{9826r(r+1)(r+2)(34r^2-43r+34)}\\
 &\times\Bigl(
 46818(r-5)^5+2342991(r-5)^4\\
 &\hspace{12mm}
 +32339210(r-5)^3+190996038(r-5)^2\\
 &\hspace{12mm}
 +516933880(r-5)+528093063
 \Bigr),
 \end{aligned}
\]
which is positive.  This proves the lemma.
\end{proof}

\paragraph*{\bf Declaration of generative AI use}
All mathematical constructions in Section~\ref{sec:constructions}
and the graph model in Section~\ref{sec:nonexistence} were developed
entirely by the authors without AI assistance.
The authors used ChatGPT 5.6 pro to assist in the improvement of the upper bounds
for the expression on the right-hand side of
\eqref{eq:key_idea} in Subsection \ref{sec:coefficient} and part of the computations in Subsection \ref{sec:uniform-proof}.
The authors also used ChatGPT 5.6 pro to improve the English
expression and readability of the manuscript.
All AI-assisted mathematical arguments have been independently
checked by the authors, and all AI-assisted text has been reviewed
and edited by them.
The authors take full responsibility for the correctness,
originality, and integrity of the entire paper.

\end{document}